\documentclass[11pt,reqno]{amsart}
\usepackage[utf8]{inputenc}
\usepackage{amssymb, MnSymbol, bbm}
\usepackage{amsmath, amsthm}
\usepackage{epsfig, enumerate}
\usepackage{color}

\let\OLDthebibliography\thebibliography
\renewcommand\thebibliography[1]{
  \OLDthebibliography{#1}
  \setlength{\parskip}{0pt}
  \setlength{\itemsep}{0.2pt plus 0.3ex}
}

\newtheorem{theo}{Theorem}
\newtheorem{lemma}[theo]{Lemma}

\newtheorem{proposition}[theo]{Proposition}
\numberwithin{theo}{section}
\numberwithin{equation}{section}

\newtheorem{remark}[theo]{Remark}

\newtheorem*{assumption}{Assumption}

\theoremstyle{definition}
\newtheorem{definition}[theo]{Definition}

\newcommand{\bfu}{{\mathbf u}}

\newcommand{\Aa}{{\mathcal A}}
\newcommand{\Bb}{{\mathcal B}}

\newcommand{\Ee}{{\mathcal E}}

\newcommand{\Nn}{{\mathcal N}}
\newcommand{\Oo}{{\mathcal O}}
\newcommand{\Pp}{{\mathcal P}}
\newcommand{\Qq}{{\mathcal Q}}

\newcommand{\Ss}{{\mathcal S}}
\newcommand{\Tt}{{\mathcal T}}

\newcommand{\Vv}{{\mathcal V}}
\newcommand{\Ww}{{\mathcal W}}
\newcommand{\Xx}{{\mathcal X}}
\newcommand{\Yy}{{\mathcal Y}}
\newcommand{\Zz}{{\mathcal Z}}

\newcommand{\Ab}{{\mathbb A}}

\newcommand{\CC}{{\mathbb C}}

\newcommand{\EE}{{\mathbb E}}

\newcommand{\GG}{{\mathbb G}}

\newcommand{\II}{{\mathbb I}}

\newcommand{\KK}{{\mathbb K}}

\newcommand{\NN}{{\mathbb N}}

\newcommand{\PP}{{\mathbb P}}

\newcommand{\RR}{{\mathbb R}}

\newcommand{\VV}{{\mathbb V}}

\newcommand{\XX}{{\mathbb X}}

\newcommand{\ZZ}{{\mathbb Z}}

\newcommand{\nul}{{\bf 0}}

\newcommand{\qtx}[1]{\quad\text{#1}\quad}

\newcommand{\Sym}{{\rm Sym}}

\newcommand{\pmat}[1]{\begin{pmatrix} #1  \end{pmatrix}}
\newcommand{\smat}[1]{\left( \begin{smallmatrix} #1  \end{smallmatrix} \right)}

\newcommand{\diag}{{\rm diag}}
\DeclareMathOperator{\im}{{\rm Im}}

\DeclareMathOperator{\ran}{{\rm ran}}
\DeclareMathOperator{\Tr}{{\rm Tr}}

\DeclareMathOperator{\spec}{{\rm spec}}

\title[GOE statistics for Anderson models]{GOE statistics for Anderson models on antitrees and thin boxes in $\ZZ^3$ with deformed Laplacian}

\author{Christian Sadel}
\thanks{This research was funded by the Chilean FONDECYT 1161651 and by the Iniciativa Científica Milenio through the Núcleo Mileneo NC120062 on 'Stochastic Models of Complex and DIsordered Systems'.}
\address{Facultad de Matem\'aitcas, Pontificia Universidad Cat\'olica de Chile} 
\email{chsadel@mat.uc.cl}


\begin{document}

\begin{abstract}
 Sequences of certain finite graphs - special types of antitrees - are constructed along which the Anderson model shows GOE statistics, i.e. a re-scaled eigenvalue process converges to the ${\rm Sine}_1$ process.
The Anderson model on the graph is a random matrix being the sum of the adjacency matrix and a random diagonal matrix with independent identically distributed entries along the diagonal.  The strength of the randomness stays fixed, there is no re-scaling with matrix size.
These considered random matrices giving GOE statistics can also be viewed as random 
Schr\"odinger operators $\Pp\Delta+\Vv$ on thin finite boxes in $\ZZ^3$ where the Laplacian $\Delta$ is deformed by a projection $\Pp$ commuting with $\Delta$.
\end{abstract}

\maketitle

\tableofcontents

\section{Introduction}

In the theory of randomly disordered systems there are two very important fields of research, the Anderson model introduced in \cite{And} and random matrix ensembles such as the Gaussian Orthogonal Ensemble (GOE) introduced by Wigner \cite{Wi}.
The latter one model the observed repulsion between energy levels (eigenvalues) in large nuclei. This is characterized by the local eigenvalue statistics which for the GOE is given by the ${\rm Sine}_1$ process in the limit where the matrix size goes to infinity, see e.g. \cite{Me}. This type of limiting statistics is expected for a wide range of disordered systems of the same symmetry class which is referred to as a universal behavior or simply {\it universality}. The GOE statistics applies to models with time reversal symmetry in delocalized regimes.
Without time reversal symmetry (for instance in presence of magnetic phases) disordered systems are expected to follow the local statistics of the Gaussian Unitary Ensemble (GUE) given by the ${\rm Sine}_2$ process. 

Let us give some small introduction to Random Matrix Theory. For more detailed information we refer to the common literature, e.g. \cite{Me, AGZ,  For}.
The Gaussian Orthogonal Ensemble ${\rm GOE}=({\rm GOE}(N))_{N\in \NN}$ is given by the collection of Gaussian distributions ${\rm GOE}(N)$ with density proportional to 
$e^{-\frac N4 \Tr(H^2) } $ on the set of {\it real symmetric} 
$N \times N$ matrices $\Sym(N)$ for $N\in \NN$. The distribution ${\rm GOE}(N)$ is invariant under orthogonal conjugation, $O^\top {\rm GOE}(N) O\stackrel{d}{=}{\rm GOE}(N)$ for $
O\in{\rm O}(N)=\{O\in {\rm Mat}(N,\RR)\,:\,O^\top O=\II_N\}$, where
$\II_N$ denotes the $N\times N$ unit matrix.
If the random matrix $H_N$ is drawn from GOE$(N)$ this means that it has {\it real} entries,  $(H_N)_{j,k}$ for $j\geq k$ (entries above and on the diagonal) are independent Gaussian random variables with mean zero and variance $1/N$ 
for the off-diagonal variables and variance $2/N$ for the diagonal. The other entries are determined as $H_N$ is symmetric.
Similar, the Gaussian Unitary Ensemble is defined in a way that the distribution is Gaussian and invariant under unitary conjugation.

The normalization in $N$ assures that in the limit $N\to \infty$ the (random) density of states measures $\nu_N=\frac{1}{N} \sum_{i=1}^N \delta_{\lambda^N_i}$ converge to the Wigner semi-circle measure $\frac{1}{2\pi}\sqrt{4-\lambda^2}\,d\lambda$ supported on $[-2,2]$.
Here, $\{\lambda^N_i\,:\,i=1,\ldots,N\}$ is the (random) set of eigenvalues of the random matrix $H_N$ and $\delta_{\lambda} $ denotes the delta measure supported at $\lambda$.
The distribution of these eigenvalues for finite $N$ follow the $\beta$-ensemble rule, that is, the symmetrized distribution of $(\lambda^N_1,\ldots,\lambda^N_N) \in \RR^N$ is proportional to
$e^{-\beta N/4 \,\sum_i \lambda_i^2} \prod_{i<j} |\lambda_i-\lambda_j|^\beta \prod_i d\lambda_i$ where $\beta=1$ for GOE and $\beta=2$ for GUE. 

The ${\rm Sine}_\beta$ processes emerge as limits of the local statistics of these joint distributions.
For some value in the so called bulk spectrum $\lambda_0 \in (-2,2)$ one can consider the eigenvalue process around $\lambda_0$, that is the shifted eigenvalue process $\spec(H_N-\lambda_0\II_N)=\{\lambda^N_i-\lambda_0\,:\,i=1,\ldots,N\}$.
According to the semi-circle law, the number of eigenvalues in a fixed small neighborhood around $\lambda_0$ is proportional to $N\sqrt{4-\lambda_0^2}$
and the distance of $\lambda_0$ to the next eigenvalue is roughly proportional to $1/(N\sqrt{4-\lambda_0^2})$. So in order to get a limiting point process it is reasonable to consider $\Sigma_N\,:=\,N \sqrt{4-\lambda_0^2}/(2\pi)\,\spec(H_N-\lambda_0\II)$. We view this random discrete set as a random counting measure $\sigma_N=\sum_{x\in \Sigma_N} \delta_x$, that is $\sigma_N(A)=|\Sigma_N\cap A|$ for $A\subset \RR$. Thus, we have a probability distribution $\Nn_N$ on the set of discrete counting measures on $\RR$. In general such a distribution is called a point process. 
$\Nn_N$ converges weakly, $\Nn_N \Rightarrow {\rm Sine}_1$ (or $\Sigma_N\Rightarrow {\rm Sine}_1$) for $N\to \infty$.
Weak convergence of point processes is given by convergence of the Laplace functional $\Psi_{\Nn_N}(f) \to \Psi_\Nn(f)$ for non-negative continuous functions $f$ with compact support, where
$$
\Psi_\Nn(f)\,:=\,\EE_\Nn(\exp(-\sigma(f)))\,=\,\int d\Nn(\sigma)(\exp(-\sigma(f))\,.
$$
In essence, this functional replaces the role of the characteristic function (Fourier transform) for probability distributions on $\RR$.
Now, in the considered case, the processes $\Nn_N$ and ${\rm Sine}_\beta$ are also uniquely characterized by their moment measures (or joint intensities) which all exist and one obtains vague-convergence of those. Under certain growth conditions 
vague convergence of finite moment measures is sufficient for obtaining weak convergence\footnote{This is similar as the comparison of weak convergence vs. convergence of moments for probability distributions on $\RR$.} which is the case here.
For a point process $\Nn$ the finite moment measures are given by the expectations over the power measures, this means for  a bounded Borel set $A_1\times A_2\times\cdots \times A_k \subset \RR^k$ the $k$-th finite moment measure is given by
$$
\EE_\Nn \left(\sigma^k(\prod_{i=1}^k A_i)\,\right)\,=\, \int d\Nn(\sigma) \prod_{i=1}^k \sigma(A_i)\;.
$$
${\rm Sine}_1$ is a so called Pfaffian process where all the finite moment measures are absolutely continuous and their densities are given by a certain Pfaffian.
If $H_N$ where drawn from the GUE we would have convergence to the ${\rm Sine}_2$ process 
which is a determinantal process given by the famous Sine-kernel, 
$K(x,y)=\frac{\sin(\pi(x-y))}{\pi(x-y)}$, giving the name Sine processes. 
This means, the finite moment measures are given by
$$
\int d{\rm Sine}_2(\sigma) (\sigma^k(f))\,=\,\int f(x_1,\ldots,x_k)\,\det[K(x_i,x_j)]_{1\leq i \leq k, 1 \leq j \leq k}\;dx_1,dx_2\cdots dx_k\;.
$$

Universality (limiting ${\rm Sine}_\beta$ process) has been proved for many random matrix ensembles, e.g. \cite{DG, ESY, TV} and particularly very recent works by Ajanki, Erd\H{o}s, Kr\"uger and Schr\"oder allow very general profiles of covariance structures and dependencies with slow correlation decay of the random entries \cite{AEK, EKS}. However, these ensembles are still  far from ensembles of very sparse matrices (many zero entries) or matrices with randomness only along the diagonal, both of which apply for Anderson models. \\

The Anderson model is supposed to describe the quantum motion of electrons in randomly disordered solids like doped semi-conductors. Unlike the random matrix ensembles here one considers operators on an infinite dimensional separable Hilbert space. Typically it is given by $\ell^2(\ZZ^d)$ or $\ell^2(\GG)$ for some countable graph $\GG$ and the random entries just appear on the diagonal. This means that one considers a random operator $H=\Delta+\Vv$ given by the sum of a real random diagonal  multiplication operator $\Vv$ (in the canonical basis) with independent identically distributed entries along the diagonal and a graph-Laplacian or adjacency operator $\Delta$. The physically most relevant models are given by the sum of a random potential and the discrete Laplacian on $\ZZ^d$, $d=1,2,3$.  There are also continuous analogues defining Anderson models on $L^2(\RR^d)$ where $\Delta$ is the actual Laplacian on $\RR^d$ and $\Vv$ a random multiplication operator made out of a sum of bump like potentials centered around lattice points and multiplied by i.i.d. real random variables.

For Anderson models one can also consider eigenvalue statistics if one restricts the model to sequences of finite cubic boxes in $\ZZ^d$ or adequate finite sub-graphs of $\GG$ approaching the infinite graph. Restricting the Anderson model to a finite box gives a random matrix. However, the sequences of such random matrices are very different from the random matrix ensembles mentioned above. The random entries are only on the diagonal and the variances are constant and not re-scaled with the size of the matrix. The off-diagonal entries are typically very sparse\footnote{meaning most off diagonal entries are zero, the non-zero entries are sparse} describing the graph structure (edges, edge-weights).

For one and quasi-one dimensional models, e.g. \cite{GMP,KuS, KLS} and for large disorder or at band edges in any dimension, e.g. \cite{FS,AM,DLS,Klo} the Anderson model localizes. This means one has pure point spectrum and exponentially localized eigenfunctions, a phenomenon called Anderson localization. 
In regimes of Anderson localization one finds Poisson type statistics (i.e. limiting Poisson point processes) \cite{Mi, Wa, GK}.

While there is a huge literature on Anderson localization, there is still a major open problem concerning delocalization.
For Anderson models of low disorder on $\ZZ^d$ for dimension $d\geq3$ it is expected that some absolutely continuous spectrum persists\footnote{The non-disordered Laplacian has purely absolutely continuous spectrum}. Moreover, in these delocalized regimes one also expects some form of universality (GOE statistics) for the eigenvalue statistics along increasing boxes approaching $\ZZ^d$.
However, so far even the existence of this delocalized regime for models on $\ZZ^d$ is mathematically unproven.

Existence of a delocalized regime for Anderson models was first shown for infinite dimensional regular trees (Bethe lattice) and then extended to similar tree like structures and tree-strips
\cite{Kl, ASW, FHS, KLW, AW, KS, Sa-TS, Sa-Fib}. Only recently some examples of graphs with finite $d$-dimensional growth rate ($d>2$) have been introduced with rigorous proofs of absolutely continuous spectrum for Anderson models on them. These are so called antitrees and similar graph structures \cite{Sa-AT, Sa-OC}. 
The word antitree describes that these graphs are far from trees as they have some local complete-graph-like structures in them.  These can be viewed as local mean-field structures which give a local averaging effect on the random potentials preventing localization.

A connection from Anderson models to GOE statistics has been found by considering long strips within $\ZZ^2$ and re-scaling the random potential in relation to the graph size. Originally one also had to modify the Laplacian slightly \cite{VV} which was later resolved in \cite{SV}.
 In this paper we combine methods from \cite{Sa-AT} with \cite{VV,SV} to construct examples of sequences of Anderson models on finite graphs with {\rm fixed disorder strength} that show GOE statistic in the limit. An additional re-scaling of the randomness (in relation to the non-random parts) as in \cite{SV,VV} is not needed.
 The graphs are tensor products of two-dimensional grids and a complete graph with normalized edge weights. As tensor products of such a complete graph with the line $\ZZ$ or half line $\ZZ_+$ are special cases of antitrees as described in \cite{Sa-AT}, we call these graphs antitrees as well.
The locally averaging graph structure of the complete graph part replaces the re-scaling of the randomness in \cite{VV,SV}. 
In some sense this sequence of considered models lies in between the theory of random band matrix ensembles and the Anderson models on $\ZZ^d$.

\subsection{The considered graphs and related random matrices}

Let us introduce more precisely the graph structures we will consider.
\begin{definition}
a) A discrete weighted graph $(\GG,W)$ is a countable 
or finite set $\GG$ together with a symmetric, real valued weight function $W:\GG\times \GG \to \RR$.
 Two distinct points $x\neq y \in \GG$ are considered to be connected by an edge if and only if $W(x,y)\neq 0$ in which case $W(x,y)=W(y,x)\in\RR$ is the edge weight.
 The diagonal elements $W(x,x)$ will be referred to as point weights. One may think of $W$ as a real symmetric matrix indexed by points in $\GG$.
 This is the adjacency matrix of the weighted graph.\\
b) The complete graph of $s$-elements with re-normalized edge weights $(\KK_s,P_s)$ is given by
$\KK_s\,:=\,\{1,\ldots,s\}$, $P_s(j,k)=\frac1s$ for any $j,k\in \KK_s$, thus any point is connected to any other point and the weights are normalized by $\frac{1}{s}$. With this normalization, $P_s$ can be viewed as a rank one orthogonal projection and thus $\|P_s\|=1$ independent of $s$.\\
c) If $(\GG,W)$ is a discrete weighted graph, then the $\GG$-antitree of constant width $s$ is given by the tensor product 
$(\GG,W)\otimes(\KK_s,P_s)=(\GG\times \KK_s, W\otimes P_s)$ where
$W\otimes P_s((x,j),(y,k))=W(x,y) P_s(j,k)=\frac{1}{s} W(x,y)$.
\end{definition}
In part c), the $\GG$-antitree of constant width is basically obtained by replacing any vertex $x\in \GG$ by a set 
$S_x$ of $s$ vertices and the edges between $x$ and $y$ 
by $s^2$ edges connecting all points in $S_x$ and $S_y$. Doing this procedure where $|S_x|$ is not constant we would get a general $\GG$-antitree (of non-constant width $|S_x|$).
The antitrees we worked with in \cite{Sa-AT} would all be $\ZZ_+$-antitrees in this sense where $\ZZ_+$ is the half line of positive integers with edges only between neighbors. Therefore, we also use the term antitree here as well.

We will consider such antitrees of constant width (tensor products with $\KK_s$) for (long) two-dimensional strips.
Such adjacent matrices can also be obtained through some deformation of the Laplacian of 3-dimensional (thin) boxes as we shall see.
More precisely, the $n\times r$ strip $\ZZ_{n\times r}$ with point weight $w$ is the set $\{1,\ldots,n\}\times\{1,\ldots,r\}\subset \ZZ^2$ with weight function
$$
W_w(x,y)=\begin{cases}
        0 & \;\text{if}\; \|x-y\|_1>1 \\
        1 & \,\text{if}\,\|x-y\|_1=1\\
        w & \,\text{if}\;x=y
       \end{cases}
$$
The corresponding $\ZZ_{n\times r}$ antitree of constant width $s$ shall be denoted by $\Ab^w_{n\times r,s}$ and the corresponding $nrs \times nrs$
adjacency operator by $\Aa^w_{n\times r,s}$, i.e.
$(\Ab^w_{n\times r,s},\Aa^w_{n\times r,s})=(\ZZ_{n\times r},W_w)\otimes (\KK_s,P_s)$.
To represent it in matrix form, we will split the $nrs \times nrs$ matrix $\Aa^w_{n\times r,s}$ into $rs \times rs$ blocks and each of these blocks is split into $s\times s$ blocks.
Identifying the base space with
$$
\ZZ_{n\times r\times s}=\{(x_1,x_2,x_3)\in\ZZ^3,\; 1\leq x_1\leq n,\, 1\leq x_2\leq r,\,1\leq x_3\leq s\}\;,$$ 
$\Aa^w_{n\times r,s}$ can be considered as an operator on $\ell^2(\ZZ_{ n \times r\times s})$ and we use
the canonical orthonormal basis $(\delta_{x_1,x_2,x_3})$ with lexicographical order to represent  $\Aa^w_{n\times r,s}$ as a matrix. We start with the mean-field vector
\begin{equation}
1_s\,:=\,\frac1{\sqrt{s}}\, \pmat{1\\ \vdots \\ 1}\,\in\,\RR^s \qtx{and note that} P_s=1_s1_s^\top = \frac1s \pmat{1 & \cdots & 1 \\ \vdots & & \vdots \\ 1 & \cdots & 1}\,\in\,\RR^{s\times s}\;.
\end{equation}
Then, define the $rs \times rs$ matrices with $s\times s$ blocks
\begin{equation}\label{eq-def-A}
P_{r,s}:=\II_r \otimes P_s=\pmat{P_s \\ & \ddots \\ & & P_s} \qtx{and}
A^w_{r,s}:=\pmat{wP_s & \;P_s\; & \\ \; P_s\; & \ddots & \ddots\\ & \ddots & \ddots & \;P_s\;\\
& & \;P_s\; & wP_s}\;.
\end{equation} 
Here, $\II_r$ is the $r\times r$ identity matrix and $A^w_{r,s}$ is block-tri-diagonal with $s\times s$ blocks. Finally, we have 
$\Aa^w_{n\times r,s}$ as a block-tri-diagonal $nrs\times nrs$ matrix structured in $rs\times rs$ blocks:
\begin{equation}
\Aa^w_{n\times r,s}\,=\,\pmat{A^{w}_{r,s} & P_{r,s} \\ P_{r,s} & A^w_{r,s} & \ddots \\  & \ddots & \ddots & P_{r,s} \\ & & P_{r,s} & A^w_{r,s}}
\end{equation}

Let be given a probability distribution $\nu$ on $\RR$.
The {\it Anderson type model} on $\Aa^w_{m,\times r}$ with single site  distribution $\nu$ is given by the random real symmetric matrix
\begin{equation}\label{eq-def-H}
H^w_{n,r,s}:=\Aa^w_{n\times r,s}\,+\,\Vv_{nrs}
\end{equation}
where $\Vv_{nrs}$ is a $nrs\times nrs$ real diagonal matrix with independent identically $\nu$-distributed random variables along the diagonal. We will assume that the distribution is compactly supported.
This is a family of random band-matrices with randomness only on the diagonal, size $N=nrs$, band-width $2rs$ and sparse structure in the entries, but with some local mean-field setup within groups of $s\times s$ blocks.

\subsection{Main results}

\begin{assumption}
{\rm (A1)}
We assume that the distribution $\nu$ of the single site potential (diagonal entries of $\Vv_{nrs}$) is compactly supported, say in the interval $[-\sigma, \sigma]$.\\
{\rm (A2)} Furthermore let us assume that the distribution is centered, $\EE(v):=\int v\,d\nu(v)=0$, and non-trivial
$\EE(v^2)=\int v^2\,d\nu(v)>0$.
\end{assumption}
 
We will often need averaged quantities over the distribution $\nu$ or products thereof $\nu^{\otimes k}$ 
as the diagonal entries of $\Vv_{nrs}$ are all independently $\nu$ distributed.
When the random variables and their dependence on these entries are clear, we will denote the expectations values by $\EE$. Furthermore, in these expressions a variable $v$ will express a $\nu$-distributed independent random variable.
 
Now let us introduce the harmonic mean
\begin{equation}
h_\lambda\,:=\, \left(\EE\left(\frac1{\lambda-v}\right)\,\right)^{-1}\,=\,
\left(\int \frac1{\lambda-v} \,d\nu(v)\,\right)^{-1}
\end{equation}
and define the interval
\begin{equation}
I_{w,\nu}\,:=\,\{\lambda\in \RR\,:\, |\lambda|>\sigma\qtx{and} |h_\lambda-w|<4\}\,.
\end{equation}
The harmonic-mean to arithmetic-mean inequality gives $|h_\lambda| < |\lambda|$ for $|\lambda|>\sigma$
and we find
\begin{equation}
[-4+w,4+w]\,\setminus\,[-\sigma,\sigma]\;\subset\;I_{w,\nu}\;.
\end{equation}
So for small $\sigma$ or large $|w|$ the set $I_{w,\nu}$ will not be empty.

\begin{theo} \label{th-main}
Let $H^w_{n,r,s}$ be the Anderson models on the antitree $\Ab^w_{n\times r,s}$ with single site distribution $\nu$ under the assumptions {\rm (A1)} and {\rm (A2)}.
 For almost any $\lambda\in I_{w,\nu}$ there exist sequences $ s_k \gg n_k \gg r_k \to \infty$, and normalization constants $\Nn_k$ such that 
 $$
\Nn_k \spec(H^w_{n_k,r_k,s_k}-\lambda)\,\Rightarrow\;\;{\rm Sine}_1 \quad\text{for $k\to\infty$}
 $$
 The growth of $s_k/n_k \to \infty$ and $r_k\to\infty$ can be chosen as slow as one wants,
meaning that for any increasing function $f(n)$ growing towards infinity one finds sequences $s_k,n_k,r_k$ satisfying this limit with $s_k/n_k < f(n_k)$ and $r_k<f(n_k)$.
\end{theo}

\begin{remark}
The spectrum $\spec(H^w_{n_k,r_k,s_k}-\lambda)$, i.e. the eigenvalues of  $H^w_{n_k,r_k,s_k}-\lambda$ are considered as a random point process and the convergence holds in the sense of a weak limit of random point processes as described in the introduction for the GOE case.

\end{remark}

Let us go back to blocks in $\ZZ^3$ and consider the set $\ZZ_{ n \times r \times s}$ as introduced above, a $n\times r \times s$ grid within $\ZZ^3$.
Now let us introduce the discrete Laplacian $\Delta_{n,r,s}$ on $\ZZ_{n \times r \times s}$ 
but with periodic boundary conditions in the last coordinate direction. For the other directions we use Dirichlet boundary conditions. This corresponds to introducing an additional edge from points $(x_1,x_2,1)$ to 
$(x_1,x_2,s)$ for any $x_1,x_2$.  All the edges get weight one and we have no point weights.
This means the matrix (or weight function) associated to $\Delta_{n,r,s}$ is given by 
$$
\langle \delta_x,\,\Delta_{n,r,s} \,\delta_y\rangle\,=\Delta_{m,n,r}(x,y)=\begin{cases} 
 1 & \;\text{if}\; \|x-y\|_1=1 \\
        1 & \,\text{if}\;\; \{x_3,y_3\}=\{1,s\} \;\text{and}\;(x_1,x_2)=(y_1,y_2)\\
        0 & \,\text{else}
\end{cases}
$$
Using the same basis structure as before we obtain
$$
\Delta_{n,r,s}=\pmat{\Delta_{r,s} & \II_{rs} \\ \II_{rs} & \Delta_{r,s} & \ddots \\ & \ddots & \ddots & \II_{rs} \\ & & \II_{rs} & \Delta_{r,s} } 
$$
where in general $\II_{m}$ will denote the $m\times m$ identity matrix and $\Delta_{r,s}$ is an $rs\times rs$ tri-diagonal block matrix made of $s\times s$ blocks given by
$$
\Delta_{r,s}=\pmat{\Delta_{s}^p & \II_{s} \\ \II_{s} & \Delta^p_{s} & \ddots \\ & \ddots & \ddots & \II_{rs} \\ & & \II_{s} & \Delta^p_{s}}
\qtx{where} 
\Delta^p_s=\pmat{0 & 1 & 0 & & 1 \\ 1 & \ddots & \ddots & \ddots \\ 0 & \ddots & \ddots & \ddots & 0 \\ & \ddots & \ddots & \ddots & 1 \\
1 & & 0 & 1 & 0}\,\in\,\RR^{s\times s}\;.
$$
Because of the periodic boundary condition in the third coordinate we use the superscript 'p' for $\Delta_s^p$. This periodicity is reflected by the top-right and bottom-left entry $1$ in $\Delta^p_s$.
This Laplacian commutes with the orthogonal projection $\Pp$ onto the functions which are constant along the third coordinate direction, i.e.
$$
\Pp\Delta_{n,r,s}=\Delta_{n,r,s} \Pp \qtx{where} \Pp \psi(x_1,x_2,x_3) = \frac{1}{s} \sum_{k=1}^s \psi(x_1,x_2,k)\,.
$$
In matrix form using $s\times s$ blocks we have the block structure $\Pp=\diag(P_{s},\ldots,P_{s})$ with $nr$ such blocks.
Using $\Delta^p_s 1_s=2\,\cdot\,1_s $ implying $P_s \Delta_s^p = \Delta_s^p P_s = 2 P_s$ and $\II_s P_s=P_s\II_s=P_s$ we find
$$
\Pp \Delta_{n,r,s} \,=\, \Pp \Delta_{n,r,s} \Pp = \Aa^{2}_{n\times r,s}\;.
$$
Here, the $\Aa^2$ is not the the square of $\Aa$, rather in our notation as above it means that we take 
the adjacency matrix $\Aa^w_{n \times r,s}$ of the $\ZZ_{n\times r}$-antitree with the point weight $w=2$ on $\ZZ_{n\times r}$. 
This leads to the following corollary:

\begin{theo}
For almost all $\lambda\in I_{2,\nu}$, in particular almost all energies $\lambda\in [-2,-\sigma)\cup(\sigma,6]$,
there are sequences $s_k \gg n_k \gg r_k \to \infty$ ( $r_k\to\infty,\, s_k/n_k \to \infty$ can grow as slow as wanted in comparison to the growth of $n_k$) such that with the correct normalization $\Nn_k$ we find
$$
\Nn_k\, \spec(\Pp\Delta_{m_k,n_k,r_k}\,+\,\Vv_{m_kn_kr_k}-\lambda)\,\Rightarrow\, {\rm Sine}_1\;.
$$
\end{theo}

\begin{remark}
{\rm (i)} The Laplacian $\Delta$ on $\ZZ^3$ or $\NN^3$ can be seen as some sort of limit of $\Delta_{n,r,s}$ for $n,r,s\to\infty$ 
and has spectrum $[-6,6]$. Indeed, for any $n,r,s\in\NN$ we have $\spec\Delta_{n,r,s}\subset[-6,6]$. 
However the projection $\Pp$ reduces to the subspace with top energy for the Laplacian in the $x_3$ direction leading to $\spec \Pp \Delta_{n,r,s}\subset [-2,6]$ and for $n,r,s\to\infty$ one fills this interval.\\[.2cm]
\noindent {\rm (ii)} Modifying the proofs slightly one can replace $\Delta_{n,r,s}$ by the Dirichlet-Laplacian $\Delta_{n,r,s}^D$ on the $n\times r\times s$ grid. However, since $\Delta^D_{n,r,s}$ does not commute with $\Pp$ one needs to consider
$\Pp \Delta^D_{n,r,s} \Pp+\Vv_{nrs}$ and a similar calculation as above shows $\Pp \Delta^D_{n,r,s} \Pp\,=\,\Aa^{2-2/s}_{n\times r,s}$. So the whole difference is a further $s$-dependence which in the limit $s\to\infty$ has no influence. This resembles the fact that in the limit towards infinity the boundary conditions of the Laplacian do not matter.\\[.2cm]
\noindent {\rm (iii)} This result can not be seen as a limiting statistics for boxes on some fixed Anderson model on a separable, infinite dimensional Hilbert space because there is no limit of the projections $\Pp=\Pp(n,r,s)$ in $\ell^2(\ZZ^3)$ and  there is no operator limit of $\Pp \Delta_{n,r,s}$ for boxes of size $n\times r \times s$ approaching $\ZZ^3$. \\[.2cm]
\noindent {\rm (iv)} With $s_k/n_k$  and $r_k$ growing very slowly, less than any power of $n_k$, the corresponding subset $\ZZ_{n_k,r_k,s_k}$ of $\ZZ^3$ look like very thin rectangular shaped boxes in $\ZZ^3$.
\end{remark}

\section{Transfer matrices}
The block structure of $\Aa^w_{n\times r,s}$ will allow the analysis of the eigenvalue equation through transfer matrices of similar type as in
 \cite{Sa-AT}. In order to see this, let us identify $\psi\in\CC^{nrs}=(\CC^{rs})^n$ with
$(\psi_i)_{i=1}^n$, $\psi_i\in \CC^{rs}$ and  let $\psi_0=\vec 0=\psi_{n+1}$.
Moreover, let us write the random diagonal matrix $\Vv_{nrs}$ in diagonal block form of $n$ blocks of size $rs$, 
$\Vv_{nrs}=\diag(V_1,V_2,\ldots V_n)$. Then, we find
$$
\left(H^w_{n, r,s} \psi\right)_i\,=\,
P_{r,s} (\psi_{i+1}+\psi_{i-1})\,+\,(A^w_{r,s}+V_i) \,\psi_i\;.
$$
By definition of the random diagonal matrix $\Vv_{nrs}$, the 
$V_i$ are i.i.d. random diagonal $rs\times rs$ matrices where each $V_i$ has diagonal entries that are all i.i.d. real random variables and $\nu$-distributed.
The projection $P_{r,s}$ can be written as
\begin{equation} \label{eq-def-Phi}
P_{r,s}=\Phi_{r,s} \Phi_{r,s}^\top \qtx{where} \Phi_{r,s}=\pmat{1_s & & 0 \\ & \ddots \\ 0 & & 1_s} \,\in\, \RR^{rs\, \times\, r}
\end{equation}
is a $rs \times r$ matrix. Recall that $1_s\in \CC^s$ is the normalized 'mean-field column vector' $1_s=1/\sqrt{s} \,(1,1,\ldots,1)^\top$.  Note that $\Phi_{r,s}^\top\Phi_{r,s}=\II_r$.

For $\psi=(\psi_i)_{i=1}^n \in \CC^{nrs}$, $\psi_i\in\CC^{rs}$ let us define
\begin{equation}
\vec{u}_i=\vec{u}_i(\psi):=\Phi_{r,s}^\top \psi_i\,\in\,\CC^r\;,
\end{equation}
then, the eigenvalue equation $H^w_{n,r,s} \psi= z\psi$ gives
\begin{equation}\label{eq-eig1}
(z-A^w_{r,s}-V_i)\psi_i \,=\, \Phi_{r,s}\,(\vec{u}_{i+1}\,+\,\vec{u}_{i-1})\;.
\end{equation}
For $z\not\in \spec(A^w_{r,s}-V_i))$ it follows that
\begin{equation}\label{eq-eig2}
\vec{u}_i\,=\, \Phi_{r,s}^\top\,(z-A^w_{r,s}-V_i)^{-1}\,\Phi_{r,s}\,(\vec{u}_{i+1}\,+\,\vec{u}_{i-1})\;.
\end{equation}
For $z\in\CC, \;\im(z)>0$  we have $\Im(z-A^w_{r,s}-V_i)^{-1}<0$, where $\Im(A)=(A-A^*)/(2\imath)$ is the imaginary part in the $C^*$ algebra sense.
Using that $\Phi_{r,s}$ is injective we get
$\im(v^*  \Phi_{r,s}^\top\,(z-A^w_{r,s}-V_i)^{-1}\,\Phi_{r,s}v)<0$ for non zero vectors $v$ and therefore, $\Phi_{r,s}^\top\,(z-A^w_{r,s}-V_i)^{-1}\,\Phi_{r,s}$ is invertible. 
Hence, it is defined and invertible for all but finitely many values of $z\in \RR$ as the determinant is a rational function of $z$. If it is invertible we can re-write the eigenvalue equation in the following form, 
\begin{equation} \label{eq-trans}
\pmat{\vec{u}_{i+1} \\ \vec{u}_{i}}\,=\,
T^{w,z}_{i;r,s}\,\pmat{\vec{u}_{i} \\ \vec{u}_{i-1}}\;,\quad
T^{w,z}_{i;r,s}\,:=\,\pmat{\left(\Phi_{r,s}^\top\,(z \II_{rs}-A^w_{r,s}-V_i)^{-1}\,\Phi_{r,s}\right)^{-1} & -\II_r \\ \II_r & \nul} \;.
\end{equation}
We call $T^{w,z}_{i;r,s}$ the $i$-th transfer matrix at energy $z$ of $\Aa^w_{n\times r,s}$.
We write the energy or spectral parameter $z$ as an upper index because the dependence on $z$ is somewhat of the same flavor as the one on $w$. 

Now let $Q_s$ be an $s\times (s-1)$ matrix such that $(1_s,Q_s)$ is orthogonal, meaning that
$$
1_s 1_s^\top\,+\,Q_s Q_s^\top\,=\,\II_s\qtx{} 1_s^\top Q_s=\nul\;.
$$
Then, we let
$$
Q_{r,s}\,:=\,\pmat{Q_s \\ & \ddots \\ & & Q_s} \qtx{implying}
\Phi_{r,s}\Phi_{r,s}^\top+Q_{r,s} Q_{r,s}^\top=\II_{rs}\;.
$$
Hence, $(\Phi_{r,s}, Q_{r,s})$ is orthogonal. 
If $M$ is an invertible $rs\times rs$ matrix where
$\Phi_{r,s}^\top M \Phi_{r,s}$ is also invertible then the Schur complement formula gives that
\begin{align}
& \left(\Phi_{r,s}^\top M^{-1} \Phi_{r,s}\right)^{-1}\,=\,
\Phi_{r,s}^\top\,M\,\Phi_{r,s}\,-\,\Phi_{r,s}^\top\,M\, Q_{r,s} \,\left(Q_{r,s}^\top\,M\, Q_{r,s} \right)^{-1} Q_{r,s}^\top\,M\, \Phi_{r,s}\;.
\end{align}
Applying this to $M=z \II_{rs}-A^w_{r,s}-V_k$ and using
$P_s Q_s = \nul $, $A_{r,s}^w Q_{r,s}= \nul$
we obtain
\begin{align}
& \left(\Phi_{r,s}^\top\,(z \II_{rs}-A^w_{r,s}-V_i)^{-1}\,\Phi_{r,s}\right)^{-1}\,=\,
-\Phi_{r,s}^\top A^w_{r,s} \Phi_{r,s}\,+\,\left(\Phi_{r,s}^\top\left(z\II_{rs}-V_i \right)^{-1} \Phi_{r,s}\right)^{-1}\;.
\end{align}
Using $1_s^\top P_s 1_s =1$, \eqref{eq-def-A} and \eqref{eq-def-Phi} we get
$$
\Phi_{r,s}^\top A^w_{r,s} \Phi_{r,s}\,=\,\Delta^D_r\,+\, w \,\II_r
$$
where $\Delta^D_r$ is the Dirichlet Laplacian on the line $\ZZ_r=\{1,\ldots,r\}$ which is slightly different from the 
periodic one $\Delta^p_s$ used above,
$$
\Delta^D_r\,:=\,\pmat{0 & 1 \\ 1 & \ddots & \ddots \\ & \ddots& \ddots  & 1\\ & & 1 & 0}\;.
$$

The diagonal matrix $V_k$ can be further partitioned into $s\times s$ blocks to obtain
\begin{equation}
\left(\Phi_{r,s}^\top (z\II_{rs}-V_i)^{-1} \Phi_{r,s}\right)^{-1}=:V_{i;s}^z=\pmat{v^z_{i,1;s} \\ & \ddots \\ & & v^z_{i,r,s}}
\end{equation}
where 
\begin{equation} \label{eq-v-fin}
v^z_{i,j;s}\,:=\,\left(1_s^\top (z\II_s-V_{i,j})^{-1} 1_s\right)^{-1}\,=\,
\left( \frac1s \sum_{k=1}^s \frac{1}{z-v_{i,j,k}}\right)^{-1}
\end{equation}
with $v_{i,j,k}$ being the random potential at the point $(i,j,k)$ so that
\begin{equation}
V_i=\pmat{V_{i,1} \\ & \ddots \\ & & V_{i,r}}\;\in\,\RR^{rs\times rs}\;\qtx{with}
V_{i,j}=\pmat{v_{i,j,1} \\ & \ddots \\ & & v_{i,j,s}} \in \RR^{s\times s} \;.
\label{eq-Vi-parts}
\end{equation}

Therefore we finally obtain 
\begin{equation}\label{eq-T-fin}
T^{w,z}_{i;r,s}\,=\,\pmat{V^z_i-w\II_r-\Delta^D_r & -\II_r \\ \II_r & \nul}\;.
\end{equation}
For some parameters $z=\lambda\in \RR$ some of the inverses in the definition of the transfer matrix \eqref{eq-trans} are not defined.
However, whenever possible we define it by analytic extension of the map $z\mapsto T^{w,z}_{i;r,s}$. 
Note that by definiteness of the imaginary parts in the occurring inverses there is never a problem for non-real parameters $z\not\in\RR$.  For this reason we define:
\begin{definition}
The value $\lambda\in\RR$ (spectral parameter) is called singular for $H^w_{n,r,s}$ at the $i$-th slice if the map
$z\mapsto T^{w,z}_{i;r,s}$ is not defined in $\lambda$ after analytic extensions.  
We call $\lambda\in \RR$ singular for $H^w_{n,r,s}$ if it is singular at some slice $i=1,\ldots,n$. 
\end{definition}
Note that by \eqref{eq-v-fin} and \eqref{eq-T-fin}  the finite set of singular parameters for $H^w_{n,r,s}$ is contained
in the convex hull of the support of $\nu$ and hence inside  the interval $[-\sigma,\sigma]$.

\section{The spectrum}

For the spectrum and the determination of singular energies we may first split off some (trivial) part of the 
matrix $H^w_{n,r,s}$. For calculating the appearing Schur complement in the $i$-th transfer matrix
it is sufficient to consider the subspace
\begin{equation}
\VV_i\,:=\, {\rm span}\left[\;\bigcup_{k=0}^\infty\,\ran\,\left( (A^w_{r,s}+V_i)^k\,\Phi_{r,s}\,\right)\;\right]
\end{equation}
which is the union of all cyclic spaces of $A^w_{r,s}+V_i$ associated to the column vectors of $\Phi_{r,s}$.
It is clear that $A^w_{r,s}+V_i$ leaves the (random) subspace $\VV_i$ and its orthogonal complement $\VV_i^\perp$ invariant. 
Now, writing $\psi\in \CC^{nrs}$ as $(\psi_i)_{i=1}^n$ with $\psi_i\in \CC^{rs}$ we use the fact that
$
\CC^{nrs}\,\cong\, \prod_{i=1}^n \CC^{rs}
$
and with this isomorphy we can identify the product $\VV$ of the $\VV_i$ as subspace of $\CC^{nrs}$ and we also have a natural embedding $\hat \VV_i^\perp$ of the complements $\VV_i^\perp$ into $\CC^{nrs}$, 
\begin{equation}
\VV\,:=\,\prod_{i=1}^n \VV_i\,\subset\,\CC^{nrs} \qtx{and} 
\hat\VV_i^\perp\,:=\, \prod_{j=1}^{i-1} \{0\}\,\times\, \VV_i^\perp\,\times\,\prod_{j=i+1}^n \{0\}\;.
\end{equation}
This means
$\psi\in\VV\,\Leftrightarrow\,\forall i=1,\ldots,n\,:\, \psi_i\in\VV_i$ and $
\psi\in\hat\VV_i^\perp\,\Leftrightarrow\, \left( \psi_j=0  \;\;\text{for}\;\;j\neq i \;\;\text{and}\;\; \psi_i\in\VV_i^\perp\,\right)$.
We should mention that it is possible that $\VV_i^\perp=\{0\}$ for all $i$ and $\VV=\CC^{nrs}$ is the full space.
In fact, for continuous distributions $\nu$ of the single-site potentials this will happen with probability one. 

\begin{proposition}
We find including multiplicities that
$$
\spec(H^w_{n,r,s})\,=\, \spec(H^w_{n,r,s}|\VV)\,\cup\,\bigcup_{i=1}^n \spec(V_i\,|\,\VV_i^\perp)\;.
$$
where $\VV_i^\perp$ is non-trivial and $\spec(V_i\,|\,\VV_i^\perp)$ non-empty
if and only if there is $j\in\{1,\ldots,r\}$ such that $V_{i,j}$
has a multiple eigenvalue with $V_{i,j}$ as defined in \eqref{eq-Vi-parts}.
\end{proposition}

\begin{proof}
Since $\ran \Phi_{r,s} \in \VV_i$ we see that $H^w_{n,r,s}$ leaves $\VV$ for any $i=1,\ldots,n$ invariant. 
Similarly, for any $\psi_i\in \VV_i^\perp$ we have $P_{r,s}\psi_i=\Phi_{r,s}\Phi_{r,s}^\top \psi_i=0$ giving that
$H^w_{n,r,s}$ also leaves all the spaces $\hat\VV_i^\perp$ invariant and the restrictions of
$H^w_{n,r,s}$ to $\hat\VV_i^\perp$ is isomorphic to the restrictions of $A^w_{r,s}+V_i$ to $\VV_i^\perp$. 
Now for $\psi_i \in \VV_i^\perp\in \CC^{rs}$ we can split up $\psi_i$ once more into $r$-parts $(\psi_{i,j})_{j=1}^r$ by $\CC^{rs}=(\CC^s)^r$ and use the block structure for $A^w_{r,s}$ as in \eqref{eq-def-A} and $\Phi_{r,s}$ as in \eqref{eq-def-Phi}. Then
$$
0=\Phi_{r,s}^\top\psi_i=\pmat{1_s \\ & \ddots \\ & & 1_s}^\top \pmat{\psi_{i,1} \\ \vdots \\ \psi_{i,r}}\,=\,
\pmat{1_s^\top\psi_{i,1}\\ \vdots \\ 1_s^\top \psi_{i,r}}
$$
implies $1_s^\top \psi_{i,j}=0$ for all $j=1,\ldots,r$. This in turn implies $P_s \psi_{i,j}= 1_s 1_s^\top \psi_{i,j}=0$
and from \eqref{eq-def-A} we get $A^w_{r,s}\psi_i = 0$.
Therefore we have
$$
A^w_{r,s}\,|\,\VV_i^\perp\,=\, \nul \qtx{implying} \left(A^w_{r,s}+V_i\right)\,|\,\VV_i^\perp\,=\, V_i\,|\,\VV_i^\perp\;.
$$
By the considerations above we have
$$
H^w_{n,r,s}\,\cong\, H^w_{n,r,s}\,|\,\VV\,\oplus\, \bigoplus_{i=1}^n V_i\,|\,\VV_i^\perp
$$
in terms of an orthogonal sum of operators  (in fact matrices). The spectral decomposition follows.

Moreover, by construction, $\VV_i^\perp$ is non-trivial precisely if there is a non-zero eigenvector $\psi_i$ of $A^w_{r,s}+V_i$ which is orthogonal to all column vectors of $\Phi_{r,s}$. By the calculations 
above this is equivalent to finding $j$ and $\psi_{i,j}\neq 0$ such that
$1_s^\top \psi_{i,j}=0$ and $\psi_{i,j}$ is an eigenvector of $V_{i,j}$ as defined in \eqref{eq-Vi-parts}.
Using the fact that $V_{i,j}$ is diagonal, you can find such an eigenvector precisely if $V_{i,j}$ has an eigenvalue with multiplicity more than one.
\end{proof}

Considering the eigenvalue equation \eqref{eq-eig1} and \eqref{eq-eig2} a solution $\psi=(\psi_i)_{i=1}^n$ can be obtained from a solution $(\vec u_i)_i$ of the transfer matrix equation by taking 
$\psi_i=\Psi_{z,i}\vec u_i$ where $\Psi_{z,i}$ is a $rs \times r$ matrix given by
\begin{equation}\label{eq-psi-zi}
\Psi_{z,i}\,:=\,(z\II_{rs}\,-\,A^w_{r,s}\,-\,V_i)^{-1}\,\Phi_{r,s}\,\left(\Phi_{r,s}^\top\left(z\II_{rs}-A^w_{r,s}-V_i \right)^{-1}\,\Phi_{r,s}\right)^{-1}\;.
\end{equation}
\begin{lemma}\label{lem-psi_E,i}
For any non-singular energy $\lambda\in\RR$ the matrices $\Psi_{\lambda,i}$ are defined or can be defined by analytic extension of $z \mapsto \Psi_{z,i}$ at $z=\lambda$.
\end{lemma}
\begin{proof} Let $\lambda$ be non singular for $H^w_{n,r,s}$.
If $\lambda\not \in \spec (A^w_{r,s}+V_i)$  then the statement is clear by existence of the first inverse in \eqref{eq-psi-zi} and existence of the last term at least by analytic extension in $\lambda$. So let $\lambda$ be an eigenvalue of $A^w_{r,s}+V_i$.
In order to show that $\Psi_{\lambda,i}$ is defined by analytic extension, it is sufficient to show that
$\varphi^\top \Psi_{\lambda,i}$ can be defined by analytic extension for any eigenvector $\varphi$ of  the real symmetric matrix $A^w_{r,s}+V_i$  because there is an orthonormal basis of eigenvectors.
So let $(A^w_{r,s}+V_i) \varphi=\lambda_0 \varphi$.
Then, for $\varepsilon\neq 0$, $|\varepsilon|$ small,
$$
\varphi^\top \psi_{\lambda+\varepsilon,i}\,=\, \frac{\varphi^\top \Phi_{r,s}}{\lambda+\varepsilon-\lambda_0} \left(\Phi_{r,s}^\top\left((\lambda+\varepsilon)\II_{rs}-A^w_{r,s}-V_i \right)^{-1}\,\Phi_{r,s}\right)^{-1}\,.
$$
For $\lambda\neq \lambda_0$ it is clear that the limit $\varepsilon\to 0$ exists as $\lambda$ is non singular and therefore the limit of the second term exists.
Let us now assume $\lambda=\lambda_0$. We need to use some Schur complement formulas. Since $\Phi_{r,s}^\top \Phi_{r,s}=\II_r$ we can choose some orthonormal basis for $\CC^{rs}$ and $\CC^r$ such that $\Phi_{r,s}\equiv \smat{\II \\ \nul}$. We may work in these bases and give $E\II-A^w_{r,s}-V_i$ and $\varphi$ the corresponding block structures  
$$
\lambda\,\II_{rs}-A^w_{r,s}-V_i\,\equiv\,\pmat{A & B \\ B^\top & D} \qtx{and} \varphi \equiv \pmat{\varphi_1 \\ \varphi_2}\,.
$$
The symbol $\equiv$ shall remind that this is not how the matrices are defined but their appearance after some basis change putting $\Phi_{r,s}$ in the block structure as indicated.
Then the eigenvalue equation for $\varphi$ transforms to
\begin{equation}\label{eq-rel-varphi12}
A\varphi_1+B\varphi_2=0, \quad B^\top\varphi_1 + D\varphi_2=0\;
\end{equation}
and we find using the Schur complement formula
\begin{align*}
& \left(\Phi_{r,s}^\top((\lambda+\varepsilon)\II_{rs}-A^w_{r,s}-V_i)^{-1} \Phi_{r,s}\right)^{-1} \Phi_{r,s} \varphi\,\equiv\,
(A+\varepsilon\II-B(D+\varepsilon\II)^{-1}B^\top)\varphi_1 \\
& \quad = \varepsilon \varphi_1 + A\varphi_1 + B(D+\varepsilon)^{-1} D \varphi_2\;\;
\stackrel{\varepsilon\to 0}{\longrightarrow}
A\varphi_1+B\varphi_2-BP_{ker} \varphi_2\,=\,-BP_{ker}\, \varphi_2
\end{align*}
where $P_{ker}$ is the orthogonal projection onto the kernel of $D$ (note that $D$ is self-adjoint).
We know that the limit $B(D+\varepsilon\II)^{-1} B^\top$ exists as $\lambda$ is not singular. Hence, for any vector $v$ we have that
$$
\lim_{\varepsilon\to 0} (B^\top v)^\top (D+\varepsilon \II)^{-1} B^\top v \qtx{exists which implies} B^\top v \in (\ker D)^\perp\;.
$$
Therefore $\ran B^\top\subset (\ker D)^\perp$ and $BP_{ker}=(P_{ker} B^\top)^\top=\nul$.
Hence,
$$
(A+\varepsilon \II-B(D+\varepsilon \II)^{-1}B^\top)\varphi_1\,\to\,0 \qtx{for} \varepsilon \to 0\;.
$$
This implies
\begin{align*}
 & \left(\varphi^\top\Psi_{\lambda+\bar\varepsilon,i}\right)^\top=\Psi_{\lambda+\bar \varepsilon,i}^\top\;\varphi\;\equiv\;
\frac{1}{\varepsilon} \left(A+\varepsilon\II -B(D+\varepsilon \II)^{-1} C \right)\varphi_1\,\to\,\\
&\qquad \frac{d}{dz}\,\left[\,\left(\Phi_{r,s}^\top (z\II_{rs}-A^w_{r,s}-V_i)^{-1} \Phi_{r,s}\, \right)^{-1}\,\Phi_{r,s} \varphi\,\right]_{z=\lambda}
\end{align*}
for $\varepsilon\to 0$, which exists because the extension of the Schur complement is analytic in $z=\lambda$.
\end{proof}

Let us now introduce the products of the transfer matrices:
\begin{equation}
\XX^{w,z}_{i;r,s}\,:=\, T^{w,z}_{i;r,s} \,T^{w,z}_{i-1;r,s} \,\cdots\, T^{w,z}_{2;r,s}\,T^{w,z}_{1;r,s}\,.
`\end{equation}

We finally obtain the key Proposition of this section

\begin{proposition}\label{prop-EV}
Let $\lambda\in \RR$ be non-singular for $H^w_{n,r,s}$.  Then, $\lambda$ is an eigenvalue of $H^w_{n,r,s}$ if and only if
we have either
\begin{equation}\label{eq-cond-eig}
\det\,\left( \,\pmat{\II_r & \nul}\; \XX^{w,\lambda}_{n;r,s}\;\pmat{\II_r \\ \nul}\, \right)\,=\, 0
\end{equation}
or $\lambda$ is an eigenvalue of $V_i\,|\,\VV_i^\perp$ for some $i=1,\ldots,n$.\\
Particularly, if $|\lambda|>\sigma$ then $\lambda$ is an eigenvalue of $H^w_{n,r,s}$ if and only if
\eqref{eq-cond-eig} holds.
\end{proposition}
Note that the second statement follows immediately as $\|V_i\|\leq \sigma \Rightarrow \spec(V_i)\subset [-\sigma,\sigma]$
 and all singular energies are also inside the interval $[-\sigma,\sigma]$.
\begin{proof}
First let $\lambda\in \spec(H^w_{n,r,s})$, either $\lambda\in \spec(V_i|\VV_i)$ for some $i=1,\ldots,n$ or $\lambda\in\spec(H^w_{n,r,s}|\VV)$.
For the letter case let $\psi=(\psi_i)_{i=1}^n$ be a corresponding non-zero eigenvector, note $\psi\in\VV$. \\
Claim 1: For some $i=1,\ldots,n$ we have  $\vec u_i=\Phi_{r,s}^\top\psi_i\neq \vec 0.$ 
If $\Phi_{r,s}^\top\psi_i=0$ for all $i=1,\ldots,n$, then
$H^w_{n,r,s}\psi=\lambda\psi$ also implies $V_i\psi_i=\lambda\psi_i$ and we have $\psi_i\in\VV_i^\perp$ and hence $\psi\in\VV^\perp$ implying $\psi=0$ as $\psi\in\VV$ as well.\\
Claim 2: $(\vec u_i)_i=(\Phi_{r,s}^\top\psi_i)_i$ satisfy the transfer matrix equation \eqref{eq-trans} at $z=\lambda$ with $\vec u_0=\vec u_{n+1}=\vec 0$. If all appearing inverses in the definition of $T^{w,\lambda}_{i;r,s}$ in \eqref{eq-trans} exist for all $i=1,\ldots,n$ then this is clear so we focus on the case when the transfer matrix is defined only by analytic extension.
The eigenvalue equation for $\lambda$ leads to
$$
\psi_i\,=\,((\lambda+\varepsilon)\II-A^w_{r,s}-V_i)^{-1}\,\left(\Phi_{r,s}(\vec u_{i+1}+\vec u_{i-1})\,+\,\varepsilon\,\psi_i \right)
$$
which after multiplying with $\Phi_{r,s}^\top$ from the left gives
$$
\left( \Phi_{r,s}^\top\left( (\lambda+\varepsilon)\II-A^w_{r,s}-V_i\right)^{-1} \Phi_{r,s}\right)^{-1}\vec u_i\,=\,
\vec u_{i+1}+\vec u_{i-1}\,+\,\varepsilon\,\Psi_{\lambda+\varepsilon,i}^\top\,\psi_i
$$
In both equations we have to set $\vec u_0=\vec 0$ for $i=1$ and $\vec u_{n+1}=\vec 0$ for $i=n$. 
With Lemma~\ref{lem-psi_E,i} the limit $\varepsilon\to 0$ shows that
$(\vec u_i)_i$ satisfies the transfer matrix equation with the transfer matrices defined by analytic extension to $\lambda$.

As not all of the $\vec u_i$ are zero, and $\vec u_0=\vec 0$,
we find that $\vec u_1\neq \vec 0$ and we have
\begin{equation}\label{eq-cond-eig2}
\pmat{\II_r & \nul}\; \XX^{w,\lambda}_{n;r,s}\;\pmat{\II_r \\ \nul}\,\vec u_1\,=\,\vec 0\,.
\end{equation}
This implies \eqref{eq-cond-eig}.

Conversely, assume \eqref{eq-cond-eig}, then we find $\vec u_1 \neq \vec 0$ satisfying \eqref{eq-cond-eig2}. 
We again focus on the case where one or more of the transfer matrices at $z=\lambda$ are only defined by analytic extension.
We let 
$\vec u_i\equiv \vec u_i(\lambda+\varepsilon)$ be defined by the transfer matrix equation, i.e.
$\smat{\vec u_{i+1} \\ \vec u_{i}} = \XX^{w,\lambda+\varepsilon}_{i;r,s} \smat{\vec u_1 \\ \vec 0}$. Note that we will have 
$\vec u_{n+1}(\lambda)=\vec 0$ by \eqref{eq-cond-eig2}. Let us define
$$
\psi_i(\lambda+\varepsilon)\,:=\,\Psi_{\lambda+\varepsilon,i}\,\vec u_i(\lambda+\varepsilon)\qtx{and}
\psi(\lambda+\varepsilon)=(\psi_i(\lambda+\varepsilon))_{i=1}^n\,. 
$$
For $\varepsilon\neq 0$ and small where all inverses in the definition of the transfer matrix exist we obtain
$$
\left(H^w_{n,r,s}-(\lambda+\varepsilon)\II_{nrs}\right)\,\psi(\lambda+\varepsilon)\,=\, \varphi(\lambda+\varepsilon)\,=\,(\varphi_i(\lambda+\varepsilon))_{i=1}^n
$$
where 
$$
\varphi_i(\lambda+\varepsilon)=\vec 0 \qtx{for} i=1,\ldots,n-1 \qtx{and}
\varphi_{n}(\lambda+\varepsilon)\,=\,-\Phi_{r,s}\,\vec u_{n+1}(\lambda+\varepsilon)\;.
$$
In the limit $\varepsilon\to 0$ with Lemma~\ref{lem-psi_E,i} we obtain that $\psi(\lambda)$ is a non-zero eigenvector for the eigenvalue $\lambda$.
\end{proof}

\section{Effective energy, effective potential and elliptic channels}

For $|\lambda|>\sigma$  the random variables $(v^\lambda_{i,j;s})_{i,j}$ are well defined and independent identically distributed, the distribution depends on $\lambda$ and $s$. Moreover, the law of large numbers gives for $|\lambda|>\sigma$ and $s\to\infty$ a limit distribution concentrated on the point $h_\lambda$. From \eqref{eq-T-fin} we thus define the effective energy by
\begin{equation}\label{eq-def-E}
E=E(\lambda)\,:=\, h_\lambda\,-\,w\;.
\end{equation}
Also note that $h_\lambda^{-1}=\EE(1/v^\lambda_{i,j;s})=\EE(1/(\lambda-v))$ where $v$ is a $\nu$-distributed random variable. 
Another important quantity will be the $\lambda$-dependent variance
\begin{equation}\label{eq-def-sigma}
\sigma^2_\lambda\,:=\, \int \left((\lambda-v)^{-1}-h_\lambda^{-1}\right)^2\,d\nu(v)\,=\,\EE\left(\frac{1}{\lambda-v}-\frac1{h_\lambda}\right)^2
\end{equation}
Moreover, let us define 
\begin{equation}\label{eq-defs-WY}
\frac1s W_{s}(\lambda)\,:=\,\EE(v^\lambda_{i,j;s})\,-\,h_\lambda \qtx{and} \frac{1}{\sqrt{s}}\,Y_{i,j;s}(\lambda)\,:=\, v^\lambda_{i,j;s}\,-\,h_\lambda\,-\,\frac1s \,W_s(\lambda)
\end{equation}
From now  on we mostly make considerations for a fixed $\lambda\not\in[-\sigma,\sigma]$ and will omit the $\lambda$-dependence most of the time.
Note that $\EE(Y_{i,j;s})=0$ and in $(i,j)$ we have a family of real (for $\lambda$ real), independent identically distributed random variables.
Harmonic mean estimates for bounded random variables as in Theorem~\ref{th-harmonic} give
\begin{equation}\label{eq-W}
W_s(\lambda)\,=\,h_\lambda^3\,\sigma^2_\lambda\,+\,\Oo(1/s)
\end{equation}
\begin{equation}\label{eq-Y}
\EE\left( Y_{i,j;s}^2 \right)\,=\,h_\lambda^4\,\sigma^2_\lambda\,+\, \Oo(1/s) \qtx{and} \sup_{s\in\NN} \EE\left(Y_{i,j;s}^{2n}\right)\,\leq\,C_n\;.
\end{equation}
The error bounds are uniform in $\lambda$ on compact sets outside $[-\sigma,\sigma]$ (including compact subsets of $\CC$). 

The upper left $r\times r$ block entry of the transfer matrices are given by
\begin{equation}\label{eq-T-part}
E\,\II_r\,+\,\frac{1}{\sqrt{s}} Y_{i;s}\,+\,\frac{1}{s}\,W_s\,\II_r\;-\;\Delta^D_r
\end{equation}
where 
\begin{equation}
Y_{i;s}\,=\,\diag(Y_{i,1;s},\ldots,Y_{i,r;s})\;
\end{equation} 
is the effective random potential in the $i$-th slice.
In the $s\to\infty$ limit the eigenvalues and eigenvectors of $E \II_r-\Delta^D_r$ will classify some of the asymptotic behavior of the products.
Let us note that
$
\lambda\,\in\,I_{w,\nu}$ implies $E(\lambda)\,\in\,(-4,4)\;.
$ Moreover,
$E(\lambda)$ is a continuous and strictly monotone function of $\lambda$ in $I_{w,\nu}$.
As in \cite{SV} we now separate elliptic and hyperbolic channels and diagonalize $\Delta^D_r$  by
the orthogonal matrix
$$
O_{jk}\,:=\,\sqrt{2/(r+1)}\,\sin(\pi\,jk\,/\,(r+1))\;,\quad j,k=1,\ldots,r\;.
$$
The corresponding $j$-th eigenvector of $\Delta^D_r$ corresponding to the $j$-th column vector of $O$ is given by
$$
a_j\,=\,2\,\cos(\pi j\,/\,(r+1))\;,\quad j=1,\ldots,r
$$
so that 
$$
O^\top\, \Delta^D_r\,O\;=\;\diag(a_1,\ldots,a_r)\;.
$$
We focus on the case $-4<E(\lambda) \leq 0$, the other case is symmetrical.
In this case $E-a_j < 2$. In the notions of \cite{SV} we have a parabolic channel if there exists $j$ such that $|E-a_j|=2$ which in this case means $E-a_j=-2$. For any given $r$, there are $r$ such values of $E$ (and of $\lambda$). The union over $r\in\NN$ gives some countable set of values in $E$ and $\lambda$ respectively. We will omit these values.
Then, if there is no parabolic channel, there is $r_{\rm h}=r_{\rm h}(r,E)$ such that
\begin{align*}
&E-a_j\,<\,-2 &\qtx{for}& j=1,\ldots,r_{\rm h} &\qtx{(hyperbolic channels)} \\
-2 \,<\, &E-a_j \,<\, 2 &\qtx{for}& j=r_{\rm h}+1,\ldots, r &\qtx{(elliptic channels)}
\end{align*}
So we have $r_{\rm h}$ hyperbolic and $r_{\rm e}:=r-r_{\rm h}$ elliptic channels.
Note that for any fixed $E$ and $r\to\infty$, $r_{\rm e}=r_{\rm e}(r,E)$ is of the order of $r$, $r_{\rm e}\sim cr$ for some $c>0$.
Then we define $\gamma_j\in \RR$ and $z_j\in\CC$, $|z_j|=1$ by
\begin{align*}
\gamma_j+\gamma_j^{-1}\,&=\,E-a_j\,, \quad &|\gamma_j|<1&\,,& &\qtx{for} j=1,\ldots,r_{\rm h} &\\
z_j+z_j^{-1}\,&=\,E-a_{j+r_{\rm h}}\,, \quad &|z_j|=1&\,, \im(z_j)>0\,,& &\qtx{for} j=1,\ldots,r_{\rm e} &\;
\end{align*}
and as in \cite{SV} we define
\begin{equation}\label{eq-Gamma}
\Gamma\,=\,\diag(\gamma_1,\ldots,\gamma_{r_{\rm h}})\,,\quad Z=\diag(z_1,\ldots,z_{r_{\rm e}})\;
\end{equation}
as well as
\begin{equation}\label{eq-Qq}
U:=\pmat{Z^* \\ & Z}\;, \quad \tilde O:=\pmat{O \\ & O}\qtx{and} \Qq:=\pmat{\Gamma & & &  \Gamma^{-1} \\ & Z^* & Z \\ \II_{r_{\rm h}} & & & \II_{r_{\rm h}}\\ & \II_{r_{\rm e}} & \II_{r_{\rm e}} }\;.
\end{equation}
Here, $U$ is a $2r_{\rm e} \times 2r_{\rm e}$ diagonal matrix, $\tilde O$ a $2r\times 2r$ orthogonal matrix written in $r\times r$ blocks 
and $\Qq$ a $2r \times 2r$ matrix  where the rows are divided in 4 blocks of sizes $r_{\rm h},\, r_{\rm e},\, r_{\rm e},\, r_{\rm h}$ and the columns in 4 blocks of sizes $r_{\rm h},\,r_{\rm e},\,r_{\rm h},\,r_{\rm e}$. All the non-zero blocks indicated above are diagonal square matrices.
These matrices depend on $\lambda$. In order to get the eigenvalue processes we will have to vary the spectral parameter around $\lambda$ but we will use these fixed $\Qq$ $\tilde O$ and $U$ to describe our basis change cf. \eqref{eq-Tt}.

But primarily let us set one more demand on the choice of $\lambda$ or better $E(\lambda)$, respectively.
For fixed $r$ the value $r_{\rm h}$ changes exactly at the points $E=E(\lambda)$ where we have some parabolic channel.
Hence,  $I(r_0):=\{\lambda\in I_{w,\nu}\,:\,r_{\rm e}(r,E(\lambda))=r_0 \}$ is a union of intervals where $r_{\rm h}$ and $r_{\rm e}$ are constant
and $Z=Z(\lambda)$ is an analytically dependent diagonal $r_0\times r_0$ matrix. 
\begin{definition}
We say that the matrix $Z=\diag(z_1,\ldots,z_{r_{\rm e}})$ with $|z_j|=1, \im(z_j)>0$
is {\bf chaotic}, if all of the following apply for all $i,\,j,\,k,\,l\in\{1,\ldots,r_{\rm e}\}$
\begin{align*}
& z_iz_jz_kz_l\,\neq\,1\,,\quad \bar z_i z_j z_k z_l \,\neq\, 1\,, \\
& \bar z_i \bar z_j z_k z_l\,\neq 1\, \qtx{unless} \{i,j\}\,=\,\{k, l\}\;.
\end{align*}
\end{definition}

Then \cite[Lemma 5.2]{SV} gives the following.
\begin{lemma}
For each $r_{0}>0$ and Lebesgue almost all $\lambda\in I(r_0)$ we find that $Z$ as defined above is chaotic and moreover for any unitary diagonal 
$r_{0} \times r_{0}$ matrix $Z_*$ there is an increasing sequence $(n_k)_k$ of integers such that 
$Z^{n_k+1}\,\to\,Z_*$ for $k\to \infty$.
\end{lemma}

So we will consider $\lambda$ and $r$ such that we have elliptic channels\footnote{this will be the case for fixed $\lambda\in I_{w,\nu}$ and $r$ big enough}, there is no parabolic channel and such that $Z$ is chaotic.

\section{The limit of thin  boxes with fixed width}

We will first look at the situation $s=mn$ with $m$ and $r$ constant and consider the eigenvalue process for $n\to \infty$.
Furthermore, we scale energy differences to $\lambda$ by $n(h_\lambda^2 \sigma_\lambda^2+1)$ (cf. \eqref{eq-h-pert}) and define
\begin{equation}\label{eq-def-l-eps}
\lambda^\varepsilon_n\,:=\,\lambda+\frac{\varepsilon}{n(h_\lambda^2\sigma^2_\lambda+1)} \qtx{so that}
E\left(\lambda^\varepsilon_n \right)\,=\,
E(\lambda)\,+\,\frac{\varepsilon}{n}\,+\,\Oo\left(\frac{\varepsilon^2}{n^2} \right)
\end{equation}
Here, the error bound is uniform for $\varepsilon/n$ varying inside a compact set so that 
$\lambda^\varepsilon_n \in I_{w,\nu}$.

Let us map out the correspondences between notations here and in \cite{SV} in order to understand the relations of the propositions.
In principle $s$ amounts to the disorder strength and $\frac{1}{s}$ corresponds to $\lambda^2$ or better to $\sigma^2 \lambda^2$ in \cite[Section~5]{SV}, $m$ amounts to $\sigma^{-2}$.
Note in particular that the use of $\lambda$ as in this paper does not correlate to the use of $\lambda$ in \cite{SV}. But $E(\lambda)$ is the more important quantity here which corresponds to $E$ in \cite{SV}, the use of $\varepsilon$ is the same.
The size of the transfer matrices $r$ here corresponds to $d$ in \cite{SV}, moreover $r_{\rm h}$ and $r_{\rm e}$  correspond to $d_h$ and $d_e$ in \cite{SV}, respectively.

\vspace{.2cm}

Since $s=mn$ from now on, we will omit the index $s$ and replace it by $m$ and $n$. Because of the different roles of $m$ and $n$ we will place the indices differently. This way notations correspond somewhat to the ones used in \cite{SV}.
Then using the definitions \eqref{eq-Gamma} and \eqref{eq-Qq} for some fixed $\lambda$ without parabolic channel such that $Z$ is chaotic we define
\begin{equation}\label{eq-Tt}
\Tt^{\varepsilon,m}_{i;r,n}\,:=\, \Qq^{-1}\, {\tilde O}^\top\,T^{w,\lambda^\varepsilon_n}_{i;r,mn}\;\tilde O\,\Qq\;.
\end{equation}
For $\varepsilon=0$ the limit $s\to \infty$ gives the non-random matrix
\begin{equation}\label{eq-Tt-channels}
\Tt_r\,=\,\lim_{s\to\infty}  \Tt^{0,m}_{i;r,s}\,=\,\pmat{\Gamma \\ & U \\ & & \Gamma^{-1}}
\end{equation}
written in blocks of sizes $r_{\rm h},\,2r_{\rm e},\,r_{\rm h}$.
Note that the upper block has the eigenvalues of size (absolute value) $<1$, the middle part the eigenvalues of size $1$ and the lower part the eigenvalues of size $>1$. Hence, when considering products the upper part is decaying, the lower part growing and the middle part stays of order 1. For the products we will look at the same basis changes and scaling and define
\begin{equation}
\Xx^{\varepsilon,m}_{i;r,n}\,:=\, \Qq^{-1}\, {\tilde O}^\top\,\XX^{w,\lambda^\varepsilon_n}_{i;r,mn}\;\tilde O\,\Qq\;=\;\Tt^{\varepsilon,m}_{i;r,n}\,\Tt^{\varepsilon,m}_{i-1;r,n}\,\cdots\,\Tt^{\varepsilon,m}_{1;r,n}\,
\end{equation}
We also have to consider the impact of the perturbation in the spectral parameter.
From \eqref{eq-defs-WY}, \eqref{eq-T-part}, \eqref{eq-def-l-eps} we obtain using $s=mn$ that
\begin{equation}\label{eq-Tt-parts}
\Tt^{\varepsilon,m}_{i;r,n}\,=\, \Tt_r\,+\,\frac{1}{\sqrt{m}}\,\frac{1}{\sqrt{n}}\,\Yy_{i;n}^{\varepsilon,m}\,+\,\left(\frac{\varepsilon}{n}\,+\,\frac{W^{\varepsilon,m}_{n}}{mn} \right)\,\Ww_r\,+\,\Oo\left(\frac{\varepsilon^2}{n^2} \right)
\end{equation}
where 
$$
\Yy^{\varepsilon,m}_{i;n}\,:=\, \Qq^{-1}\,\tilde O^\top\,\pmat{Y_{i;mn}(\lambda^{\varepsilon}_n) & \nul \\ \nul & \nul}\,\tilde O\, \Qq\,,\;\;
W^{\varepsilon,m}_n\,:=\,W_{mn}(\lambda^\varepsilon_n)\,,\;\;
\Ww_r\,:=\,\Qq^{-1}\,\pmat{\II_r & \\ & \nul}\,\Qq\;.
$$
The error term is non-random and the bound is uniform for $\varepsilon/n$ varying in compact sets where $\lambda_n^\varepsilon$ stays outside $[-\sigma,\sigma]$.
Equation \eqref{eq-Tt-parts} is in essence the analogue of \cite[equation (5.5)]{SV} where $1/\sqrt{m}$ here plays the role of $\sigma$ there. 
Some difference is that here the randomness and the drift-term have some dependence on $\varepsilon$ and $m$, however, this dependence will not matter in the limit.
Note that by construction $\EE(\Yy^{\varepsilon,m}_{i;n})=\nul$.
Using \eqref{eq-W} and \eqref{eq-Y} we find for $\varepsilon$ varying inside compact sets that 
\begin{equation}\label{eq-errors}
W^{\varepsilon, m}_n\,=\,h_\lambda^3\sigma_\lambda^2\,+\,\Oo\left(\frac{1}{mn}, \frac{\varepsilon}{n} \right)\;,\;\;
\EE\left(Y_{i,j;mn}^2(\lambda_n^\varepsilon)\right)\,=\,h_\lambda^4 \sigma_\lambda^2\,+\,\Oo\left(\frac{1}{mn}, \frac{\varepsilon}{n} \right)\;.
\end{equation}
The error terms mean that the reminder terms are bounded by 
$C(1/(mn)+|\varepsilon|/n)$ with a uniform $C$ as long as $\varepsilon/n$ stays inside some compact interval so that always $\lambda^\varepsilon_n\in I_{w,\nu}$. In particular for any compact set $K$ there is $N$ such that for $n>N$ and $\varepsilon\in K$ this bound is uniform.

Using these bounds, the moment bound in \eqref{eq-Y} and the independence of the $Y_{i,j;s}$ we see that Theorem~\ref{theo-B} is applicable towards an SDE limit for the products $\Xx^{\varepsilon,m}_{i;n}$ for fixed $\varepsilon, m$ with scaling $i\sim n$. 
More precisely, from the decomposition of $\Tt_r$ in \eqref{eq-Tt-channels}  define
$$
\Pp_{\leq 1}\,=\,\pmat{\II_{r_{\rm h}+2r_{\rm e}} \\ \nul}\,\in\,\RR^{2r\times (r_{\rm h}+2r_{\rm e})}\;,
$$
and let
$$
X^{\varepsilon,m}_{i;r,n}\,:=\, 
\pmat{\II_{r_{\rm h}} \\ & U^{-i}}
\left(\Pp_{\leq 1}^\top \left[\Xx^{\varepsilon,m}_{i;r,n}\,\Xx_0 \right]^{-1} \Pp_{\leq 1} \right)^{-1}
$$
where $\Xx_0$ is some adequate $r\times r$ matrix such that the Schur complement
$$
 X_0\,:=\, \left(\Pp_{\leq 1}^\top \Xx_0^{-1} \Pp_{\leq 1}\right)^{-1} 
$$
exists. Then, Theorem~\ref{theo-B}~(i) gives a weak limit of stochastic processes 
$$
X^{\varepsilon,m}_{\lfloor tn \rfloor\,;\,r,n}\,\;\;\stackrel{n\to\infty}{\Longrightarrow}\;\;
\pmat{\nul \\ & \Lambda_t^{\varepsilon,m}}\, X_0\qtx{for} t\,>\,0\qtx{with} \Lambda^{\varepsilon,m}_t\,\in\,\CC^{2r_{\rm e} \times 2r_{\rm e}}
$$
being some  stochastic processes with
$\Lambda^{\varepsilon,m}_0\,:=\,\II_{2r_{\rm e}}$
which for $(\varepsilon,m)$ fixed satisfy some SDE (stochastic differential equation) in $t$. A special choice of $\Xx_0$ and hence $X_0$ as in \cite{SV} is needed for proving the limiting eigenvalue statistics mentioned further below. 
The covariance structure of the matrix Brownian motions appearing can be calculated as in \cite[Proposition~5.3 and Section~5.3]{SV}, especially 
\cite[eq. (5.37)]{SV}, as we have almost the same type of random matrices here with the same elliptic and hyperbolic channels and the same diagonalization of $\Delta_r^D$. This gives the following.
\begin{proposition}\label{prop-SDE} Let $\lambda$ be such that $Z$ is chaotic.
The family of processes $\Lambda_t^{\varepsilon,m}$ satisfy SDEs in the evolution in $t$ of the form
$$
d\Lambda_t^{\varepsilon,m}\,=\, \left(\varepsilon+\frac{h_\lambda^3\sigma_\lambda^2-q}{m}\right)\,\Ss\, \pmat{\II_{r_{\rm e}} & \\ & -\II_{r_{\rm e}}}\,\Lambda_t^{\varepsilon,m}\,dt\,+\,\frac{1}{\sqrt{m}} \,\Ss\pmat{d\Aa_t & d\Bb_t \\ -d\Bb_t^* & -d\overline \Aa_t}\,\Lambda^{\varepsilon,m}_t
$$
where
$$
\Ss\,=\,\pmat{(\overline Z-Z)^{-1} \\ & (\overline Z-Z)^{-1}}\;,\quad
q\,=\,\frac{h_\lambda^4\,\sigma_\lambda^2}{r+1}\,\sum_{j=1}^{r_{\rm h}} (\gamma_j^{-1}-\gamma_j)^{-1}\;.\,
$$
$\Aa_t$ and $\Bb_t$ are independent matrix Brownian motions, $\Aa_t$ is Hermitian and $\Bb_t$ complex symmetric, i.e.
$$
\Aa_t^*=\Aa_t\;,\quad \Bb_t^\top\,=\,\Bb_t
$$
with covariance structure
$$
\EE|(\Bb_t)_{ij}|^2=\EE|(\Aa_t)|_{ij}|^2=\EE\big((\Aa_t)_{ii} (\Aa_t)_{jj} \big)\,=\,h_\lambda^4 \sigma_\lambda^2\,t\,\cdot\,\begin{cases}
\frac32   & \text{{\rm if}}\;i=j \\ 1 & \text{{\rm if}}\; i\neq j
\end{cases}
$$
All covariances which do not follow are zero.
\end{proposition}

Note, the occuring factors $h_\lambda^4\sigma_\lambda^2$  come from the variance of $Y_{i,j;s}$ as compared to the variance $1$ for the potential used in \cite{SV}. 
As in \cite{SV} for energies close to $\lambda$ all the eigenvalues of $H^w_{n,r,s}$ are given by the zeros of $\lambda'\mapsto\det(\pmat{\II_r & \nul} \XX^{w,\lambda'}_{n;t,s} \smat{\II_r \\ \nul})$ (See Proposition~\ref{prop-EV}).
Note that for $|\varepsilon|<C$ and any $C$ there is $n_0=n_0(C,\lambda)$ such that $\lambda^\varepsilon_n \in I_{w,\nu} \subset \RR \setminus [-\sigma,\sigma]$ for any $n>n_0$ 
Using the calculations in \cite[Theorem~5.4]{SV} one can change the analytic function in $\varepsilon$ characterizing the eigenvalues along adequate sub-sequences  to get another characterization of this point process 
in the limit using Theorem~\ref{theo-B}~(iii). This leads to the following.

\begin{proposition}\label{prop-EV-lim}
Let  $\Ee_{n,r,s}$ be the process of eigenvalues of $H_{n,r,s}-\lambda \II_{nrs}$ re-scaled by the factor $n(h_\lambda^2 \sigma_\lambda^2+1)$, i.e. let
$$
\Ee_{n,r,s}\,=\,n(h_\lambda^2 \sigma_\lambda^2+1)\,\spec\left(H_{n,r,s}-\lambda\,\II_{nrs} \right)\,.
$$
Fixing $r$ let $\lambda \in I_{w,\nu}$ be such that $Z$ (as defined in \eqref{eq-Gamma}) is chaotic, let $n_k$ be some  strictly increasing sequence such that $Z^{n_k+1} \to Z_*$ for $k\to \infty$. Then, $\Ee_{n_k,r,mn_k}$ converges to the zero process of the determinant of a $r_{\rm e}(\lambda) \times r_{\rm e}(\lambda)$ matrix,
$$
\Ee_{n_k,r,mn_k}\quad\Longrightarrow\quad {\rm zeros}_\varepsilon \,\det\left(
\pmat{\overline Z_* & Z_*}\,\Lambda_1^{\varepsilon,m}\,\pmat{\II_{r_{\rm h}} \\ -\II_{r_{\rm h}}}\,\right)\qtx{for}{k\to\infty}
$$
\end{proposition}

The important part here is that the SDEs can be jointly solved in $\varepsilon$ with unique analytic versions in $\varepsilon$ (distributions on the set of analytic functions, see Theorem~\ref{theo-B}~(ii)\,). 
Therefore, the random set of zeros, ${\rm zeros}_\varepsilon f(\Lambda^{\varepsilon,m}_1)=\{\varepsilon\in \CC\,:\,f(\Lambda^{\varepsilon,m}_1)=0\}$ for an analytic function $f$ is well defined (as a distribution on the set of sets) 
and makes sense as a point process if $\PP(f(\Lambda^{\varepsilon,m}_1)\equiv 0\,\forall \varepsilon \in \CC)=0$.

The factor $(h_\lambda^2\sigma_\lambda^2+1)$ occurs here because it also occurs in the perturbations $\lambda^\varepsilon_n$ of $\lambda$.
Note that with fixing $r$ and letting $s\sim n$ going to infinity of the same order we basically look at a sequence of graphs resembling a quasi-two-dimensional limit.

\section{The GOE limit}

Let us now explain how from Propositions~\ref{prop-SDE} and \ref{prop-EV-lim} one can get to the limiting GOE statistics as in \cite{SV,VV}. Formally, the first step is like a derivative of the SDE in Proposition~\ref{prop-SDE} for small $1/\sqrt{m}$ when replacing $\varepsilon$ by $\varepsilon/\sqrt{m}$ meaning that we zoom in more locally. Then in the $m\to \infty$ lots of 
(groups) of eigenvalues of this process will move to infinity and some group is left which spaces like the eigenvalues of a random matrix with Gaussian entries. 
These random matrices are almost like in the GOE ensemble, there is just a bit of a different covariance structure and some dependence.
Afterwards, the $r\to \infty$ limit will finally lead to the ${\rm Sine}_1$ process. So all together with the limit in the previous structure, it is a triple limit process leading to the GOE statistics.

Fixing $r$ we look at the process $\sqrt{m}\,(X^{\varepsilon/\sqrt{m},m}_{i;r,n}- X_0)$ in a $m\to\infty$ limit.
On the level of the limiting process $\Lambda^{\varepsilon,m}_t$ as in Proposition~\ref{prop-SDE} let us note that
$$
\widehat \Lambda^{\varepsilon,m}_t\,:=\,\sqrt{m}\,\left(\Lambda^{\varepsilon/\sqrt{m},m}_t\,-\,\II_{2r_{\rm e}} \right)
$$
satisfies the SDE
\begin{align*}
d\widehat \Lambda^{\varepsilon,m}_t\,=\,& \left(\varepsilon-\frac{h_\lambda^2\sigma_\lambda^2-q}{\sqrt{m}}\right)\,\Ss\,
\pmat{\II_{r_{\rm e}} \\ & -\II_{r_{\rm e}}}\,\left(\frac{\widehat \Lambda^{\varepsilon,m}_t}{\sqrt{m}} \,+\,\II_{2r_{\rm e}}\right)\,dt \\
&+\, \Ss\,\pmat{d\Aa_t & d \Bb_t \\ -d\Bb_t^* & -d\overline \Aa_t}\,\left(\frac{\widehat \Lambda^{\varepsilon,m}_t}{\sqrt{m}} \,+\,\II_{2r_{\rm e}}\right)\;\qtx{with} \widehat \Lambda^{\varepsilon,m}_0\,=\,\nul\;.
\end{align*}
In the limit $m\to \infty$ the SDE can be easily solved and one finds as in \cite{SV}
$$
\widehat \Lambda^{\varepsilon,m}_t\;\;\stackrel{m\to\infty}{\Longrightarrow}\;\; \Lambda^\varepsilon_t\,:=\,\varepsilon\,t\,\Ss\,\pmat{\II_{r_{\rm e}} \\ & -\II_{r_{\rm e}}}\,+\,\Ss\,\pmat{\Aa_t & \Bb_t \\ -\Bb^*_t & \overline \Aa_t}\;.
$$
Now taking $\lambda$  such that $Z$ is chaotic as in Proposition~\ref{prop-EV-lim} and taking a sequence $n_k$ with $Z^{n_k+1} \to \II_{r_{\rm e}}$  we find the limiting eigenvalue processes 
$$
\Ee_{n_k,r,mn_k}\;\;\stackrel{k\to\infty}\Longrightarrow\;\; \Ee_{r,m}\,:=\,{\rm zeros_\varepsilon}\,\det\left(\pmat{\II_{r_{\rm e}} & \II_{r_{\rm e}}}\,\Lambda^{\varepsilon,m}_1\,\pmat{\II_{r_{\rm e}} \\ - \II_{r_{\rm e}}} \right)\;.
$$
Then, working with analytic versions in $\varepsilon$ and $1/\sqrt{m}$ for this family of processes one finds as in \cite{SV}
$$
\sqrt{m}\,\Ee_{r,m}\;\;\stackrel{m\to\infty}{\Longrightarrow}\;\; \Ee_r\;:=\;{\rm zeros}_\varepsilon\,\det\left(\pmat{\II_{r_{\rm e}} & \II_{r_{\rm e}}}\,\Lambda^{\varepsilon}_1\,\pmat{\II_{r_{\rm e}} \\ - \II_{r_{\rm e}}}  \right)\;=\;\spec\,\Re{\rm e}(\Bb_1-\Aa_1)
$$
Using the calculations as in \cite[Lemma 5.5]{SV} in combination with Theorem~\ref{theo-B}~(iii) one can get to this limit with a double sequence $n_k\gg m_k\to \infty$, more precisely:
\begin{proposition}
\label{prop-EV-lim2}
Let $\lambda$ be such that $Z$ is chaotic, let $n_k$ be a strictly  increasing sequence of natural numbers such that $Z^{n_k+1} \to \II_{r_{\rm e}}$ and let $m_k\to \infty$ be some increasing sequence towards infinity such that $\sqrt{m_k} \,\| Z^{n_k+1}-\II_{r_{\rm e}}\|\,\to\,0$. Then for $t>0$, jointly in $t\in(0,1]$ and $\varepsilon$ varying in any finite subset of $\CC$ we find
$$
\sqrt{m_k}\,\left(X^{\frac\varepsilon{\sqrt{m_k}}, m_k}_{\lfloor tn_k\rfloor;r,n_k}\,-\,X_0 \right)\quad\stackrel{k\to\infty}{\Longrightarrow}\quad \pmat{\nul \\ &  \Lambda^\varepsilon_t} \,X_0\;.
$$
Moreover, for the re-scaled eigenvalue process $\Ee_{n,r,s}$ as defined above we find
$$
\sqrt{m_k}\,\Ee_{n_k,r,m_k n_k}\;\;\stackrel{k\to\infty}{\Longrightarrow}\;\; \Ee_r\,=\,\spec \Re{\rm e}(\Bb_1-\Aa_1)\;.
$$
\end{proposition}
Let us note that from the process it is obvious that given any (slowly) towards $\infty$ increasing function $f(n)$ one can choose to consider only sequences such that $m_k<f(n_k)$.

\begin{proof}[Proof of Theorem~\ref{th-main}]
Let $b$ be some standard Gaussian variable and $K=K(r_{\rm e})$ be an independent real symmetric $r_{\rm e} \times r_{\rm e}$ matrix with Gaussian entries such that $\EE((K_{ii})^2)=\frac{5}{4}$ and $\EE((K_{ij}^2)=1$ for $i\neq j$.
Then in distribution,
$$
\Re{\rm e}(\Bb_1-\Aa_1)\;
\stackrel{d}{=}\;\frac{h_\lambda^2 \sigma_\lambda}{\sqrt{r+1}}\,\left( K\,+\, b\,\II_{r_{\rm e}}\right)
$$
As explained in \cite{VV}, using methods of \cite{ESYY} the local eigenvalue process converges to the ${\rm Sine}_1$ process for $r_{\rm e}\to \infty$, more precisely,
$$
\sqrt{r_{\rm e}}\,\spec(K(r_{\rm e})\,+\,b\,\II_{r_{\rm e}})\;\;\stackrel{r_{\rm e}\to\infty}{\Longrightarrow}\;\; {\rm Sine}_1\;.
$$
Now, for almost all $\lambda\in I_{w,\nu}$ i.e. almost all $E(\lambda)\in(-4,4)$ we find that for all $r\in\NN$, $Z$ is chaotic and there is no parabolic channel. Let us fix such a $\lambda$. 
Then for $r\to\infty$ we also find $r_{\rm e}(r,E)\to \infty$  and hence
$$
\frac{\sqrt{(r+1)r_{\rm e}}}{h_\lambda^2 \sigma_\lambda}\,\Ee_r\;\;\Longrightarrow\;\;{\rm Sine}_1\;.
$$
This convergence and the convergence mentioned in Proposition~\ref{prop-EV-lim2}
happen in the topology of weak convergence for point processes.
Therefore, one can construct some diagonal sequence $m_k,n_k,r_k\to\infty$ such that with $s_k=m_k n_k$ and $r_{{\rm e},k}=r_{\rm e}(r_k,E)$ we find
$$
\frac{\sqrt{m_k(r_k+1)\,r_{e,k}}}{h_\lambda^2 \sigma_\lambda}\;\Ee_{n_k,r_k,s_k}\;\;\stackrel{k\to\infty}{\Longrightarrow}\;\; {\rm Sine}_1\;.
$$
This proves Theorem~\ref{th-main} with the normalization constant
$$
\Nn_k\,:=\,\frac{(h_\lambda^2\sigma_\lambda^2+1)\,\sqrt{n_k s_k (r_k+1)r_{{\rm e},k}}}{h_\lambda^2 \sigma_\lambda}\;.
$$
Now let $f(n)$ be any (slowly) increasing function with $f(n)\to\infty$ for $n\to\infty$.
In Proposition~\ref{prop-EV-lim2} one may choose $m_k<f(n_k)$ and start the sequence with $n_k>f(r)$ . Therefore, we may choose $m_k<f(n_k)$ and $r_k<f(n_k)$.
\end{proof}

\appendix

\section{Harmonic means of random variables}

In the transfer matrices we see effective potentials that are harmonic means of certain independent identically distributed (iid) random variables. Certain estimates are crucial for the proofs.
We therefore consider in this section independent identically distributed random variables $X_k\in [a,b]$, $0<a<b$, $k\in \NN$.
These variables correspond to $E-v_{i,j,k}$.
We will consider the harmonic means $V_s$ and the harmonic average $h$ defined by
$$
V_s \,:=\, \frac1{\frac1s \sum_{k=1}^n \frac1{X_k}}\;,\quad
h\,:=\,\frac1{\EE(1/X_k)}\;
$$
where $\EE$ denotes the expectation value.
$V_s$ corresponds to the random variables $v^\lambda_{i,j;s}$ as in \eqref{eq-v-fin} and $h$ corresponds to $h_\lambda$.
The second and third moment of the centered random variable $1/X_j - 1/h$ will be of some importance, therefore let
$$
\sigma_m^m:=\EE((1/X_k - 1/h)^m)\;.
$$
Note $\sigma_1=0$ and $\sigma_2^2$ is the variance of $1/X_k$ and corresponds to $\sigma^2_\lambda$ in the application of the following estimates.
\begin{theo} \label{th-harmonic}
There exists a continuous function $C=C(a,b,h,\sigma_2,\sigma_3)$ such that uniformly in $s$,
\begin{equation}\label{eq-EE-M-1}
 0< \EE(V_s-h) \,\leq\, \frac{b\,h^2\, \sigma_2^2}{s}\;,\qquad
 \left| \EE(V_s-h)- \frac{h^3 \sigma_2^2}{s}\right|\,\leq\,\frac{C}{s^2}\;.
\end{equation}
\begin{equation}\label{eq-EE-M-2}
\frac{a^2\,h^2\,\sigma_2}{s}\,\leq\,\EE((V_s-h)^2)\,\leq\,\frac{b^2\,h^2\,\sigma_2^2}{s}
\,,\qquad
\left|\,\EE((V_s-h)^2)-\frac{h^4 \sigma_2^2}{s}\,\right|\,\leq\, \frac{C}{s^2}\;.\qquad
 \end{equation}
Moreover, for the higher moments we find
\begin{equation}\label{eq-EE-M-3}
 \left|\,\EE((V_s-h)^3)\,\right|\,\leq\, \frac{C}{s^2}\,,\qtx{and}
 \EE((V_s-h)^{2m})\,\leq\,\frac{(2m)!\,h^{2m}\,b^{2m} }{2^m \,m!\,a^{2m}}\,\frac{1}{s^m}\quad\text{for $m\geq 2$}\,.
\end{equation}
\end{theo}

\begin{proof}
Using $b>a>0$ we find $V_s \in [a,b]$ and $V_s\leq\frac1s \sum_{k=1} X_k$ as well as $h<\EE(X_k)$ by the arithmetic-harmonic mean inequality.
 Let $Y=1/V_s-1/h=\frac1s \sum_{k=1}^s (\frac{1}{X_k} - 1/h)$, then  $\EE(Y)=0$,\,
 $\EE(Y^2)=\sigma_2^2 / s$ and $\EE(Y^3)=\sigma_3^3/s^2$. 
 Moreover, with $Y_k:= 1/X_k-1/h$ we find $|Y_k|\leq \frac1a-\frac1b\leq \frac1a$ and
 $$
 \EE(Y^{2m})\,=\,\frac{1}{s^{2m}} \sum_{k_1,\ldots,k_{2m}=1}^s \EE(Y_{k_1}\cdots Y_{i_{km}})
\;\leq\; \frac{1}{s^{2m}} \frac{(2m)!}{2^m m!} \sum_{k_1,\ldots,k_m=1}^s \EE(Y_{k_1}^2 \ldots Y_{k_m}^2))\;
 $$
where we used that unpaired indices lead to zero expectation and the fact that $\frac{(2m)!}{2^mm!}$ is the number of pairings of the set $\{1,\ldots,2m\}$. Now using that there are $s^m$ $m$-tuples $(k_1,\ldots,k_m)$ and using the bound of $Y_k$ as mentioned above we find
$$
\EE(Y^{2m})\,\leq\, \frac{1}{s^{m}} \frac{(2m)!}{2^m m!}\,\frac{1}{a^{2m}} \;.
$$
Expanding $V_s=h-hYV_s$ repeatedly we obtain
\begin{equation}\label{eq-EE-M}
 V_s-h = -hYV_s = -h^2Y+h^2Y^2V_s=-h^2Y+h^3Y^2-h^3Y^3V_s\;.
\end{equation}
As $V_s\in[a,b]$ we can estimate $\EE(h^2Y^2 V_s)\leq h^2 b\,\sigma_2^2/s$ and
$$
\left| \EE(Y^3V_s)\right|\,\leq\, \left|\EE(hY^3)\right|+\left|\EE(hY^4V_s) \right|\,\leq\,
\frac{h(|\sigma_3^3|\,+\,3ba^{-4})}{s^2}
$$
which with \eqref{eq-EE-M} (using the second-last and last term) gives \eqref{eq-EE-M-1}.
Taking powers of \eqref{eq-EE-M} and using similar estimates lead to \eqref{eq-EE-M-2} and 
\eqref{eq-EE-M-3}.

For the general moment bound we use $V_s-h=h^2Y(YV_s-1)$ from the expansion above. 
Since \mbox{$1-YV_s$}$=V_s/h \in [a/h,b/h]$ we have$|1-YV_s|\leq b/h$ and therefore,
$$
\EE((V_s-h)^{2m})\,\leq\, (b/h)^{2m}\, \EE((h^2 Y)^{2m})\,=\, h^{2m}\, b^{2m}\, \EE(Y^{2m})\;\leq\; \frac{1}{s^{m}} \frac{(2m)!}{2^m m!}\,\frac{h^{2m}b^{2m}}{a^{2m}}
$$
\end{proof}

When varying the spectral parameter we also need to understand how the harmonic average varies for the definition in \eqref{eq-def-l-eps}. 
This amounts to replacing $X_k$ by $X_{k,\varepsilon}=X_k+\varepsilon$ and recalculating $h_\varepsilon=1/\EE(X_{k,
\varepsilon}^{-1})$.
Note by the continuity of $C=C(a,b,h,\sigma_2,\sigma_3)$ for the formulas above the error terms will also be uniform in $\varepsilon$ along compact sets $|\varepsilon|\leq c$ in $\varepsilon$ as long as $c<a$ because $X_{k,\varepsilon}\in[a-c,b+c]$ under such perturbations. Using
$$
\frac{1}{X_{k,\varepsilon}}\,=\,\frac{1}{X_k+\varepsilon}\,=\,\frac{1}{X_k}\,-\,\frac{\varepsilon}{X_k^2}\,+\,\frac{\varepsilon^2}{X_k^2 X_{k,\varepsilon}}
$$
as well as $\EE(1/X_k^2)\,=\,\sigma_2^2+1/h^2$ and defining $C_\varepsilon:=\,\EE(\frac1{X_k^2 X_{k,\varepsilon}})$ we find
\begin{equation}\label{eq-h-pert}
h_\varepsilon\,=\, \frac{1}{\frac1h\,-\,\varepsilon[\sigma_2^2+\frac1{h^2}-C_\varepsilon \varepsilon]}\,=\,
h\,+\,\varepsilon\,(\sigma^2_2\,h^2\,+\,1)\,+\,\Oo(\varepsilon^2)
\end{equation}
where the error bound is uniform on compact sets in $|\varepsilon|\leq c$ where $a-c>0$.

\section{SDE limits for products of random matrices}

In this appendix we sumerize the key results of \cite{SV} which are used in this paper.
Let be given some probability space $(\Omega,\Aa,\PP)$, an open ball of radius $r$ around zero $B_r=\{z\in\CC\,:\,|z|<r\}$ and a family of analytic random matrices
$T^\varepsilon_{k;n}:\Omega \to \CC^{r \times r}$ for $k,n \in \NN$, $\varepsilon/n \in B_r$ of the form
$$
T^\varepsilon_{k;n}\,=\,T_0\,+\, \frac{1}{\sqrt{n}} \Vv_{k;n}\,+\, \frac{1}{n} \left( \varepsilon \Yy_n+ \Ww_{n} \right)\,+\,\frac{1}{n^{3/2}}\Zz^\varepsilon_{k;n}
$$
where $(\Vv_{k;n}, \Zz^\varepsilon_{k;n})_{k=0}^\infty$ are independent identically distributed random variables (for fixed $\varepsilon$ and $n$)
and $T_0$, $\Yy_n$ and $\Ww_n$ are non-random (that is they are fixed for all $\omega\in \Omega$).
Analyticity means that for any $\omega\in\Omega$ the dependence of $T^\varepsilon_{k;n}(\omega)$ and thus of $\Zz^\varepsilon_{k;n}(\omega)$ on $\varepsilon\in nB_r=B_{nr}$ is analytic. 
The lowest order term shall be block-diagonalized in the form
$$
T_0\,=\,\pmat{\Gamma_0 \\ & U \\ & & \Gamma_2^{-1}} \qtx{with} 
\Gamma_0\in\CC^{r_0 \times r_0}\,, \;U\in {\rm U}(r_1)\,,\; \Gamma_2\in \CC^{r_2\times r_2}\;,
$$
where $\|\Gamma_0\|<1\;,\;\;\|\Gamma_2\|<1$. Here, ${\rm U}(r_1)$ is the unitary group of $r_1 \times r_1$ matrices.
Moreover, we assume that $\Vv_{k;n}$ have mean zero, $\EE(\Vv_{k;n})=0$, and that we have 
uniformly for $\varepsilon/n \in B_r$, $n\in\NN$ a 8th moment bound\footnote{in fact a $6+\delta$ moment for $\delta>0$ is enough, but here we prove this type of bound in Appendix A} in the following sense
$$
\EE(\|\Vv_{k;n} \|^8)\,<\,C \qtx{and} \EE(\|\Yy^\varepsilon_{k;n}\|^8)\,<\,C\;.
$$
Furthermore, we assume that the limits
$$
\lim_{n\to\infty} \Yy_{n}\,=\,\Yy\qtx{and} \lim_{n\to\infty} \Ww_n\,=\,\Ww
$$
exists and that we have limits of all second moments of the complex entries of $\Vv_{k;n}$ meaning that
$$
\lim_{n\to \infty} \EE\left(\Vv_{k;n}^\top M \Vv_{k;n}\right)\,=\, h(M) \qtx{and}
\lim_{n\to \infty} \EE\left(\Vv_{k;n}^* M \Vv_{k;n} \right)\,=\, \hat h(M)
$$
exist giving linear maps from $\CC^{r\times r}$ to itself.
Here, $\EE$ denotes the expectation, i.e. the integral over $\omega\in\Omega$ with respect to the probability measure $\PP$.
Without the limit $n\to\infty$ these functions encode all joint second moments of the random matrix entries of $\Vv_{k;n}$.
First, let us define some projections we will need.
$$
\Pp_{\leq 1}\,=\,\pmat{\II_{r_0+r_1} \\ \nul_{r_2\times(r_0+r_1)}} \;,\quad
\Pp_1\,=\,\pmat{\nul_{r_0\times r_1} \\  \II_{r_1} \\ \nul_{r_2 \times r_1}}\;,\quad
\Pp_2\,=\pmat{ \nul_{(r_0+r_1)\times r_2} \\ \II_{r_2}}\;.
$$
The exponential growing part for powers of $T_0$ will be projected away by a Schur complement:
Let $\Xx_0$ be such that $X_0:=(\Pp_{\leq 1}^\top \Xx_0^{-1}  \Pp_{\leq 1})^{-1}$ exists and consider
$$
X^\varepsilon_{k;n}\,:=\,\pmat{ \II_{r_0} \\ & U^{-k}}\left(\Pp_{\leq 1}^\top \,\left(
T^\varepsilon_{k;n} T^\varepsilon_{k-1;n} \cdots T^\varepsilon_{2;n} T^\varepsilon_{1;n} \Xx_0\right)^{-1} \Pp_{\leq 1} \right)^{-1}\;.
$$
The rotations through $U$ in $T_0$  lead to an averaging effect.
 The averaged covariances for a limiting Brownian motion will be described by the following functions,
$$
g(M):= \int_{\langle U \rangle} \bar \bfu \,\bar U\, \Pp_1^\top h(\Pp_1 \bfu^\top M \bfu \Pp_1^\top ) \Pp_1 U^* \bfu^*\;d\bfu
$$
$$
\hat g(M):= \int_{\langle U \rangle} \bfu \, U\, \Pp_1^\top \hat h(\Pp_1 \bfu^* M \bfu \Pp_1^\top ) \Pp_1 U^* \bfu^*\;d\bfu\;.
$$
Here, $\langle U \rangle$ is the smallest compact group containing the unitary $U$ and by the notation $d\bfu$ we integrate $\bfu\in \langle U \rangle$ 
over the normalized Haar measure on that group. Furthermore, for the drift term we define in a similar way
$$
W:=\int_{\langle U \rangle} \bfu \,\left[ \Pp_1^\top \Ww \Pp_1 \,-\ \Pp_1 h(\Pp_2 \Gamma_2 \Pp_2^\top) \Pp_1 \right]\,U^*\,\bfu^*\,d\bfu
$$
and
$$
Y:=\int_{\langle U \rangle} \bfu\,\Pp_1^\top \Yy \,\Pp_1\,U^*\,\bfu^*\,d\bfu\;.
$$
\begin{theo} \label{theo-B} {\rm (i)}
In the scaling limit $k \sim n \to \infty$ the family of processes can be described by an SDE (stochastic differential equation) in the sense that
the family of processes (for varying $\varepsilon$) converges in distribution 
$$
X^\varepsilon_{\lfloor nt \rfloor;n} \;\stackrel{}{\Longrightarrow}\; \pmat{\nul_{r_0 \times r_0} \\ & \Lambda^\varepsilon_t} \,X_0 \qtx{for}
n \to \infty\,,\; t>0\;.
$$
Here, $(\Lambda_t^\varepsilon)_{t>0}$ is a family of processes in $\CC^{r_1 \times r_1}$ satisfying an SDE in $t$ of the form
$$
{\rm d}\Lambda^\varepsilon_t\,=\,{\rm d}\Bb_t\,\Lambda^\varepsilon_t\;+\;(\varepsilon Y + W) \Lambda^\varepsilon_t \,{\rm d}t \qtx{with}
\Lambda^\varepsilon_0=\II_{r_1}\;.
$$
$\Bb_t$ is a matrix-valued Brownian motion (independent of $\varepsilon$) with covariance structure
$$
\EE(\Bb_t^\top M \Bb_t)\,=\,t\, g(M)\;,\quad 
\EE(\Bb_t^* M \Bb_t)\,=\,t\,\hat g(M)\;.
$$
{\rm (ii)} There is an analytic version of this family of processes, this means a version (same finite points distributions) such that the random functions
$\varepsilon\mapsto \Lambda^\varepsilon_t$ are analytic in $\varepsilon$.
Moreover, let $f\,:\,\CC^{(r_0+r_1) \times (r_0+r_1)} \to \CC$ be complex-analytic such that $\PP(f(\Lambda^\varepsilon_1)= 0\,\forall \varepsilon \in \CC)=0$. Then,  one has a well-defined point process
$$
{\rm zeros}_\varepsilon\left(f(\Lambda^\varepsilon_t) \right)=\{\varepsilon\in \CC\,:\,f(\Lambda^\varepsilon_1)=0 \}\,.
$$
{\rm (iii)} For some analytic function $f_0:\CC^{r\times r}\to \CC$ let be defined the point processes
$$
\Ee_n\,:=\,{\rm zeros}_\varepsilon\,f_0(T^\varepsilon_{n;n} T^\varepsilon_{n-1;n} \cdots T^\varepsilon_{1;n})
$$
which should be discrete countable sets with probability one.
Assume that one finds $\Xx_0$ as above and analytic functions $f_n:\CC^{(r_0+r_1)\times(r_0+r_1)}\to \CC$ such that
 for any compact set $K\subset \CC$ we have 
$$
\PP\left(\Ee_n \cap K \,=\, {\rm zeros}_\varepsilon f_n(X^\varepsilon_{n;n})\,\cap \,K\,\right) \to\,1\;,
$$
$f_n \to \hat f$ uniformly on $K$ and $f(\Lambda^\varepsilon):=\hat f\left( \pmat{\nul \\ & \Lambda^\varepsilon_1} X_0\right)$ fulfills the conditions of part b).
Then, in the sense of weak convergence of point processes,
$$
\Ee_n\,\Longrightarrow\, {\rm zeros}_\varepsilon \,f(\Lambda^\varepsilon_1)\;.
$$
\end{theo}
\begin{proof}
Part (i) follows directly from \cite[Theorem~1.1]{SV}. Also, note that for a finite set of $\varepsilon$, say $(\varepsilon_1,\ldots,\varepsilon_m)\in \CC^m$ we can simply consider 
block-diagonal matrices $\diag(T^{\varepsilon_1}_{k;n},\ldots,T^{\varepsilon_m}_{k;n})$ for obtaining the joint distributions for different $\varepsilon$
in the limit. This leads to the use of the same Brownian motions for different $\varepsilon$ and we have  in fact convergence
to a random field  $(\varepsilon,t)\to \Lambda^\varepsilon_t$.\\
For part (ii) first note that the limit is independent of $\Zz^\varepsilon_{k;n}$. Hence, we can set this part equal to $\nul$ first, obtaining families of random matrices $T^\varepsilon_{k;n}$ for all $\varepsilon \in \CC$ depending analytically on $\varepsilon$. 
As argued in \cite[Section~5.2]{SV} using the uniform bounds (in $\varepsilon$) one can use \cite[Corollary~15]{VV} to get a unique version for which $\varepsilon \mapsto \Lambda^\varepsilon_{t}$ is analytic (uniqueness in the sense of joint probability distributions on the set of analytic functions). For $f$ as given, one can then obtain well-defined distributions on the set of countable subsets of $\CC$ defined by the zeros of $f(\Lambda^\varepsilon_1)$ which gives a point process.\\
Part (iii) is basically proved in \cite[Theorem~5.4]{SV} for a specific case, following again \cite[Corollary~15]{VV}.  
First, for the weak convergence of point processes it is sufficient that for any compact set $K\subset \CC$ the point processes restricted to $K$ converge. 
Secondly, for $K\subset \CC$ compact and $n_0$ large enough we have $K\subset n_0 B_r$ and all 
$T^\varepsilon_{k;n}$ are defined for $\varepsilon\in K$ and $n\geq n_0$.
Moreover, using the uniform bounds and arguments in \cite{SV}, for $\omega\in \Omega_0\subset \Omega$ with $\PP(\Omega_0)=1$ we have that the Schur complements $X^\varepsilon_{k;n}$ are well defined for sufficiently large $n$ (with a possibly random lower bound).
Again, by \cite[Corollary~15]{VV} one finds analytic versions in $\varepsilon$, all realized on the same probability space, such that the convergence in part a) is uniform on compact sets (almost surely).
Thus, the zeros of $f_n(X^\varepsilon_{n;n}(\omega))$ converge to the ones of $f(\Lambda_1^\varepsilon(\omega))$ uniformly in $K$ (almost surely), 
if this limiting function is not identically zero in $\varepsilon$.
This implies the weak convergence of the point processes given by the zeros in $K$. 
\end{proof}

\end{document}